\documentclass{article}
\usepackage{a4wide,amsmath,amsthm,epsfig,graphicx}
\usepackage{amsmath,amsthm,amssymb}
\usepackage{amsfonts}
\usepackage{graphicx}
\usepackage{hyperref}
\usepackage{geometry}
\usepackage{rotating}
\usepackage{multirow}
\usepackage{subfigure}
\usepackage[T1]{fontenc}
\usepackage[sc,osf]{mathpazo}
\usepackage[english]{babel}
\title{Convenient Multiple Directions of Stratification\thanks{
This research benefited from the support of the chair "Risques
Financiers", Fondation du Risque.}}
\author{Benjamin Jourdain, Bernard Lapeyre\footnote{Université Paris-Est, CERMICS, Projet MathFi, ENPC-INRIA-UMLV,
6 et 8 avenue Blaise Pascal, 77455 Marne La Vallée, Cedex 2, France,
E-mails: jourdain@cermics.enpc.fr, bl@cermics.enpc.fr}, Piergiacomo
Sabino\footnote{Université Paris 7 Diderot, LPMA, 175 rue du
Chavaleret, 75013 Paris, France, Email: sabino@math.jussieu.fr}}
\date{}

\newtheorem{prop}{Proposition}
\newtheorem{theo}{Theorem}

\newtheorem{corr}{Corollary}
\newtheorem{rem}{Remark}

\begin{document}
\maketitle
\def\indic{1\!\!1}
\def\Var{\mathbb{V}\textrm{ar}}
\begin{abstract}
This paper investigates the use of multiple directions of
stratification as a variance reduction technique for Monte Carlo
simulations of path-dependent options driven by Gaussian vectors.
The precision of the method depends on the choice of the directions
of stratification  and the allocation rule within each strata.
Several choices have been proposed but, even if they provide
variance reduction, their implementation is computationally
intensive and not applicable to realistic payoffs, in particular not
to Asian options with barrier. Moreover, all these previously
published methods employ orthogonal directions for multiple
stratification. In this work we investigate the use of algorithms
producing convenient directions, generally non-orthogonal, combining
a lower computational cost with a comparable variance reduction. In
addition, we study the accuracy of optimal allocation in terms of
variance reduction compared to the Latin Hypercube Sampling. We
consider the directions obtained by the Linear Transformation and
the Principal Component Analysis. We introduce a new procedure based
on the Linear Approximation of the explained variance of the payoff
using the law of total variance. In addition, we exhibit a novel
algorithm that permits to correctly generate normal vectors
stratified along non-orthogonal directions. Finally, we illustrate
the efficiency of these algorithms in the computation of the price
of different path-dependent options with and without barriers in the
Black-Scholes and in the Cox-Ingersoll-Ross markets.

\end{abstract}
\section{Introduction}
The main purpose of Monte Carlo (MC) simulations is to compute
integrals numerically. It is frequently the only alternative for
solving problems in applied sciences and notably for financial
applications. The pricing of derivative contracts and value-at-risk
calculations for risk-management purposes typically require
numerical simulations.  However, the MC method for high-dimensional
problems is a demanding computational  task and a considerable
number of studies have been devoted to increase its efficiency via
variance reduction techniques. This paper investigates the use of
multiple directions of stratification as a variance reduction
technique for MC simulations of path-dependent options driven by
high-dimensional Gaussian vectors. The precision of the method
depends on the choice of the partitions of the space and the
allocation of the number of samples within each strata. Usually, the
strata are polyhedrons delimited by hyperplanes orthogonal to a few
direction vectors. Several choices have been proposed: Glasserman et
al. \cite{GHS99} select the directions for the stratification of
linear projections based on the quadratic approximation of the
integrand or payoff function. In contrast, Etoré et al.
\cite{EFJM09} find the directions by adaptive techniques. These two
approaches provide a high variance reduction but their
implementation can be computationally intensive and the former one
cannot be applied to more realistic payoff functions such as Asian
options with barrier at each time step. Moreover, these two methods
suppose orthogonal directions for multiple stratification. In this
work, we investigate the use of algorithms producing convenient
directions, generally non-orthogonal, combining a lower
computational cost with a variance reduction that is comparable to
the above mentioned methods. In addition, we study the accuracy of
optimal allocation, combined with the above stratification
techniques, in terms of variance reduction, compared to ``fixed''
allocation procedures such as Latin Hypercube Sampling (LHS). We
consider the directions produced by the Linear Transformation (LT)
decomposition introduced by Imai and Tan \cite{IT2006} and  the
Principal Component Analysis (PCA). Moreover, we propose a new
procedure based on the Linear Approximation (LA) of the
``explained'' variance of the payoff function by the use of the law
of total variance. Notably, we design a novel algorithm that permits
to correctly generate multivariate normal random vectors stratified
along non-orthogonal directions. We illustrate the efficiency of the
proposed algorithms and their combination for the computation of the
price of different path-dependent options with and without barriers
in the Black-Scholes (BS) and in the Cox-Ingersoll-Ross (CIR)
models. In the former dynamics, it turns out that the LA and the LT
approaches return the same first order direction while this vector
is almost parallel to the one obtained by the GHS technique even in
the case of Asian options with a barrier at expiry. This justifies
the application and the good performance of the LA (and LT) if the
barrier is at each monitoring time. Consequently, the approaches
return the same variance reduction and the LA (LT) is easier to
implement and has a lower computational cost. We repeat our
numerical investigation in the CIR framework where we find explicit
solutions for the LT and LA directions. In order to find a further
direction, we compute the first principal component of the sampled
covariance matrix of the price process obtained by a MC estimation
via a pilot run. In both BS and CIR dynamics, LT and LA return
remarkable variance reduction with a low computational cost. We also
show that in some setting the stratification along multiple
directions can be more efficient than stratifying along a single
one. In particular, the combination of the LA (LT) direction and a
non-orthogonal direction, notably the first principal component, can
even outperform the variance reduction of two orthogonal directions
in the case of barrier options. Finally, as far as the allocation of
the samples is concerned, in any case the LHS displays a
considerable higher computational time and has always a lower
variance reduction as compared to the use of a convenient direction
of stratification with optimal allocation.

The paper is organized as follows. Section \ref{sect:SLP} reviews
the main  ideas of stratification and the motivations of this study.
Section \ref{sect:StratNonOrth} presents the new algorithm that
permits the stratification along non-orthogonal directions. Section
\ref{sect:CD} discusses the use of convenient stratification
directions and in particular, contains the presentation of the LT
decomposition and the introduction of the LA procedure. In Section
\ref{sect:FA} we explain the financial applications and find the
explicit solutions for the LA and the LT methods both for the BS and
the CIR dynamics. In Section \ref{sect:NI} the variance reductions
and the computational costs of the proposed technique are
illustrated by numerical experiments. Finally, Section
\ref{sect:concl} concludes the paper by summarizing the most
important findings.
\section{Stratified Sampling and Linear Projections}\label{sect:SLP}
Stratified Sampling is a general variance reduction technique that
consists of drawing the observations from specific partitions of the
sample space. More specifically, suppose we want to compute by MC
simulations an expectation of the form $\mathbb{E}[g(\mathbf{Y})]$
where $g:\mathbb{R}^d\rightarrow\mathbb{R}$ is a Borel function and
$\mathbf{Y}$ is a $\mathbb{R}^d$-valued random vector with the
assumption that $\mathbb{E}[g(\mathbf{Y})^2]<\infty$. Consider a
\emph{stratification variable} $X$ and let $A_1,\dots,A_K$  be
disjoint subsets of the real line for which
$\mathbb{P}\left(\bigcup_{k=1}^K\{X\in A_k\}\right)=1$. Then
\begin{equation}
\mathbb{E}[g(\mathbf{Y})] =
\sum_{k=1}^K\mathbb{E}[g(\mathbf{Y})|X\in A_k]\mathbb{P}(X\in A_k) =
\sum_{k=1}^K\mathbb{E}[g(\mathbf{Y})|X\in A_k]p_k
\end{equation}
where $p_k = \mathbb{P}(X\in A_k),\,k=1,\dots,K$. The
\emph{stratified estimator} with $N_S$ draws is defined as:
\begin{equation}
\sum_{k=1}^Kp_k\frac{1}{n_k}\sum_{j=1}^{n_k}g(Y_{kj})=\frac{1}{N_S}\sum_{k=1}^K\frac{p_k}{q_k}\sum_{j=1}^{n_k}g(Y_{kj}),
\end{equation}
where $n_k$ are the number allocations in the $k$-th stratum and
$q_k=n_k/N_S$ is their fraction  in the $k$-th stratum and $Y_{kj}$
are independent draws from the conditional distribution of $Y$ given
$X\in  A_k$. Its variance is given by
$\sum_{k=1}^Kp_k^2\frac{\sigma_k^2}{n_k}$ where $\sigma_k$ is the
conditional variance of $g(\mathbf{Y})$ given $X\in A_k$.

This estimator may be more efficient than the usual MC sample mean
estimator of a random sample of size $N_S$. The potential higher
efficiency of the former estimator critically depends on the
allocation rule and the choice of the partition of the sample space.
The optimal allocation rule is the one that minimizes the variance
of the stratified sampling estimator given the partition of the
state space and the constraint $\sum_{k=1}^Kq_k=1$. It is given by:
\begin{equation}
q_k = \frac{p_k\sigma_k}{\sum_{k=1}^Kp_k\sigma_k}.
\end{equation}
The probabilities $p_k$ are known whereas generally the conditional
variances are not known. They can be estimated in a pilot run and
then used in a second stage to determine  the stratified estimator.
This is not the optimal procedure and more sophisticated techniques
can be employed, see for example Etoré and Jourdain \cite{EJ09}.

We focus our attention on MC simulation driven by high-dimensional
Gaussian vectors that are of particular interest in financial
applications. As such, we consider in the following only normal
random variables.

\subsection{Stratifying Linear Projections: 1-dimensional Setting}
We begin with a general description of stratifying a linear
projection of a Gaussian random vector. Suppose $\mathbf{Z}$ is a
$d$ dimensional centered Gaussian random vector,
$\mathbf{Z}\sim\mathcal{N}(\mathbf{0},\Sigma_Z)$ and then consider
the stratification variable $X$ as the linear projection of
$\mathbf{Z}$ over a fixed direction $\mathbf{v}\in\mathbb{R}^d$, $X
= \mathbf{v}\cdot\mathbf{Z}$. $X$ is also Gaussian with variance
$\mathbf{v}\cdot\Sigma_Z\mathbf{v}$. This choice permits to
partition the sample space $\mathbb{R}^d$ into strata defined by
\begin{equation}
S_{k,v}=\left\{\mathbf{x}\in\mathbb{R}^d,\,\mathbf{x}\cdot\mathbf{v}\in
A_k\right\}.
\end{equation}
Due to the Gaussian structure of the random variables we can
generate $\mathbf{Z}$ stratified along the direction $\mathbf{v}$ in
the following way. Consider a general Gaussian random vector
$\mathbf{Y}=(\mathbf{Y}_1,\mathbf{Y}_2)$:
\begin{equation}
\mathbf{Y}=(\mathbf{Y}_1,\mathbf{Y}_2)\sim\mathcal{N}\left(
  \begin{array}{cc}
  \left(
    \begin{array}{c}
      \mu_1 \\
      \mu_2 \\
    \end{array}
  \right)
     ,&
       \left(
 \begin{array}{cc}
      \Sigma_{11} & \Sigma_{12}\\
      \Sigma_{21} & \Sigma_{22} \\
    \end{array}
  \right)
  \end{array}
\right)
\end{equation}
\noindent and denote
$\mathcal{L}\left(\mathbf{Y}_1|\mathbf{Y}_2=\mathbf{x}\right)$ the
law of $\mathbf{Y}_1$ given $\mathbf{Y}_2=\mathbf{x}$, it is
possible to prove (see for instance Glasserman \cite{Glass2004})
that
\begin{equation}
\mathcal{L}(\mathbf{Y}_1\left|\right.\mathbf{Y}_2=\mathbf{x})=\mathcal{N}\left(
\mu_1+\Sigma_{12}\Sigma_{22}^{-1}\left( \mathbf{x} - \mu_2\right),
\Sigma_{11}-\Sigma_{12}\Sigma_{22}^{-1}\Sigma_{21}  \right).
\end{equation}
where we assume that $\Sigma_{22}$ is invertible. Adapting the above
result for $\mathbf{Z}$ given $X=\mathbf{v}\cdot\mathbf{Z}$ and
$\Var[X]=\mathbf{v}\cdot\Sigma_Z\mathbf{v}=1$ we have
\begin{equation}
\mathcal{L}\left(\mathbf{Z}\left|X=x\right.\right)=\mathcal{N}\left(
\frac{\Sigma_Z\mathbf{v}}{\mathbf{v}\cdot\Sigma_Z\mathbf{v}}x,\Sigma_Z-\frac{\Sigma_Z\mathbf{v}\mathbf{v}^T\Sigma_Z}
{\mathbf{v}\cdot\Sigma_Z\mathbf{v}}
\right)=\mathcal{N}\left(\Sigma_Z\mathbf{v}x,\Sigma_Z-\Sigma_Z\mathbf{v}\mathbf{v}^T\Sigma_Z
\right).
\end{equation}
If we consider $\Sigma_Z=I_d$ the above equation becomes:
\begin{equation}\label{eq:strat1dim}
\mathcal{L}\left(\mathbf{Z}\left|X=x\right.\right)=\mathcal{N}\left(\mathbf{v}x,I_d-\mathbf{v}^T\mathbf{v}
\right).
\end{equation}
The conditional covariance matrix $D = I_d-\mathbf{v}^T\mathbf{v}$
does not depend on $x$ and since $D$ is an orthogonal projection
matrix, we have  $DD^T=D$. Due to this result, we do not need to
compute the Cholesky (or a general square-root) matrix of $D$ to
sample from the conditional distribution of $\mathbf{Z}$ given $X$.
These observations give an easy and simple algorithm to generate $K$
samples of $\mathbf{Z}$ stratified along the direction $\mathbf{v}$.

Suppose now that $A_k$ is the interval between the quantiles of
order $\frac{k-1}{K}$ and of order $\frac{k}{K}$ of the standard
normal distribution. We can sample from $\mathbf{Z}$ given
$\mathbf{Z}\cdot \mathbf{v}\in A_k$ in the following steps:
\begin{enumerate}
\item generate $U\sim\mathcal{U}([0,1])$.
\item Set $V = \frac{k-U}{K}$ and $X = \Phi^{-1}(V)$,
with $\Phi$ the inverse of the cumulative normal distribution.
\item Generate
$\mathbf{Z}'\sim\mathcal{N}\left(\mathbf{0},I_d\right)$ independent
on $U$.
\item Set $\mathbf{v}X +
\left(I-\mathbf{v}\mathbf{v}^T\right)\mathbf{Z}'$.
\end{enumerate}
We suggest to implement the last term in the last step as
$\mathbf{Z}'-\mathbf{v}(\mathbf{v}\cdot\mathbf{Z}')$ which requires
$O(d)$ operation rather than $O(d^2)$.

\subsection{Stratifying Linear Projections: Multidimensional Setting}
We start with the case of orthogonal directions and consider a
matrix $V\in\mathbb{R}^{d\times d'},d'\le d$, whose columns are the
direction vectors, such that $V^TV=I_{d'}$. Following the notation
introduced above we have:
\begin{equation}
\mathbf{X} = V^T\mathbf{Z}
\end{equation}
where now $\mathbf{X}$ is $d'$ dimensional. Moreover,
\begin{equation}
  \left(
    \begin{array}{c}
      \mathbf{Z} \\
      \mathbf{X} \\
    \end{array}
  \right)
\sim\mathcal{N}\left(
  \begin{array}{cc}
  \left(
    \begin{array}{c}
      \mathbf{0} \\
      \mathbf{0} \\
    \end{array}
  \right)
     ,&
       \left(
 \begin{array}{cc}
      \Sigma_Z & \Sigma_Z V\\
      V^T\Sigma_Z & V^T\Sigma_ZV \\
    \end{array}
  \right)
  \end{array}
\right)
\end{equation}
Consequently
\begin{equation}
\mathcal{L}\left(\mathbf{Z}\left|\mathbf{X}=\mathbf{x}\right.\right)=\mathcal{N}\left(
\Sigma_ZV\left(V^T\Sigma_ZV\right)^{-1}\mathbf{x},\Sigma_Z-\Sigma_Z
V\left(V^T\Sigma_ZV\right)^{-1}V^T\Sigma_Z\right)
\end{equation}
where we assume that $V^T\Sigma_ZV$ is invertible. In the case
$\Sigma_Z=I_d$ we have
\begin{equation}
\mathcal{L}\left(\mathbf{Z}\left|X=x\right.\right)=\mathcal{N}\left(
V\left(V^TV\right)^{-1}\mathbf{x},I_d-
V\left(V^TV\right)^{-1}V^T\right).
\end{equation}
Hence, if we adopt orthogonal directions $V^TV=I_{d'}$ the algorithm
to stratify $\mathbf{Z}$ given $\mathbf{X}=V^T\mathbf{Z}$ is a
simple multidimensional version of the algorithm illustrated before
where now we should stratify the $d'$ dimensional hypercube
$[0,1]^{d'}$. Suppose, for example, that we stratify the $j$-th
coordinate of the hypercube, $j=1,\dots,d'$, into $K_j$ intervals of
equal length so that we have a total number of
$K_1\times\cdots\times K_{d'}$ equiprobable strata. In this
multidimensional setting we can sample from $\mathbf{Z}$ given
$\mathbf{X}=V^T\mathbf{Z}\in A_k$, where
$A_k=\prod_{j=1}^{d'}\Phi\left(\left[\frac{k_j-1}{K_j},\frac{k_j}{K_j}\right]\right)$,
in the following steps:
\begin{enumerate}
\item generate $\mathbf{U}=\left(U_1\dots, U_{d'}\right)$ with independent components each of law $\mathcal{U}([0,1])$.
\item Set $V_j = \frac{k_j-U_j}{K_j}$ with $j\in\{1,\dots,d'\}$ and
$k_j\in\{1,\dots,K_j\}$.
\item Set $\mathbf{X}=\left(X_1\dots, X_{d'}\right)$, $X_j = \Phi^{-1}(V_j)$.
\item Generate
$\mathbf{Z}'\sim\mathcal{N}\left(\mathbf{0},I_d\right)$ independent
of $\mathbf{U}$.
\item Set $V\mathbf{X} + \left(I_d-VV^T\right)\mathbf{Z}'$.
\end{enumerate}
We now investigate the possibility to stratify over different
directions that can be non-orthogonal either. When the directions
are not orthogonal the components of $\mathbf{X}$ are not
independent since $\Var[\mathbf{X}] =VV^T\ne I_{d'}$ and the
previous multidimensional algorithm cannot be adopted anymore. A
first way yo approach this problem may be to assume
$\mathbf{X}\stackrel{\mathcal{L}}{=}C_X\epsilon$ with
$\epsilon\sim\mathcal{N}(\mathbf{0},I_{d'})$ independent on
$\mathbf{Z}$ and $C_X\in \mathbb{R}^{d'\times d'}$ such that
$\Var[\mathbf{X}]=C_XC_X^T$, and use the following slight
modification of the above algorithm.
\begin{enumerate}
\item generate $\mathbf{U}=\left(U_1\dots, U_{d'}\right)$ with independent components each of law $\mathcal{U}([0,1])$.
\item Set $V_j = \frac{k_j-U_j}{K_j}$ with $j=\{1,\dots,d'\}$ and
$k_j=\{1,\dots,K_j\}$.
\item Set $\epsilon=\left(\epsilon_1\dots, \epsilon_{d'}\right)$, $\epsilon_j = \Phi^{-1}(V_j)$.
\item Generate
$\mathbf{Z}'\sim\mathcal{N}\left(\mathbf{0},I_d\right)$ independent
of $\mathbf{U}$.
\item Set $V(C_X^X)^{-1}\epsilon+(I_d-
V\left(V^TV\right)^{-1}V^T)\mathbf{Z}'$.
\end{enumerate}
%
However, although mathematically correct, this algorithm stratifies
the marginals of the random vector $\epsilon$ that has independent
components. This construction does not consider the fact that the
marginals of $\mathbf{X}$ are not independent and the introduction
of the dependence can affect this partial stratification in
complicated ways (see Glasserman \cite{Glass2004}).

\section{Stratification along non-orthogonal directions}\label{sect:StratNonOrth} In this
section we show how to  generate multivariate normal random vectors,
$\mathbf{Z}\sim\mathcal{N}(0,I_d)$, stratified along non-orthogonal
directions. We prove the following proposition:
\begin{prop}
Let $B_1=\{\mathbf{e}_1\,\dots,\mathbf{e}_{d'}\}$ be a set of
linearly independent vectors in $\mathbb{R}^{d'}$, $d'\le d$, such
that $\|\mathbf{e}_i\|=1,\,i=1,\dots,d'$, let $B'=
\{\mathbf{f}_1'\,\dots,\mathbf{f}'_{d'}\}$ be the set of orthogonal
vectors produced the Gram-Schmidt procedure:
$\mathbf{f}_i'=\mathbf{e}_i-\frac{\sum_{m=1}^{i-1}(\mathbf{e}_i\cdot\mathbf{f}_m')\mathbf{f}_m'}{\|\mathbf{f}_m'\|^2}$.
Finally consider $B_2 = \{\mathbf{f}_i =
\frac{\mathbf{f_i'}}{\|\mathbf{f_i'}\|},i=1,\dots,d'\}$ the
orthonormal version of $B'$ and let $F$ be the $d\times d'$ matrix
whose $i$-th column is $\mathbf{f}_i$.

Suppose $g:\mathbb{R}^d\rightarrow\mathbb{R}$ such that
$\mathbb{E}[g^2(\mathbf{Z})]<+\infty$ and consider two vectors in
$\mathbb{R}^{d'}$,
$\mathbf{a}^{\pm}=\{a_1^{\pm},\dots,a_d'^{\pm}\}$, such that
$a_i^-<a_i^+,\,\forall i=1,\dots,d'$. We have
\begin{align}\label{nonorth:eq}
&\mathbb{E}\left[g(\mathbf{Z})\left|a_i^-\le
\mathbf{Z}\cdot\mathbf{e}_i\le a_i^+,\,i=1,\dots,d'
 \right. \right] \nonumber\\
& =\mathbb{E}\left[
                g\Bigg(
                \left(I - FF^{T}\right)\mathbf{Z} \right.
                \nonumber\\
                &\left.\left.
                +\sum_{m=1}^{d'}\mathbf{f}_m\Phi^{-1}\left(
                \Phi\left(\tilde{a}_m^-\left(\mathbf{U}^{(m-1)}\right)
                                \right) +
               U_m\left(\Phi\left(\tilde{a}_m^+\left(\mathbf{U}^{(m-1)}\right)\right) -
               \Phi\left(\tilde{a}^-_m\left(\mathbf{U}^{(m-1)}\right)\right)\right)
                \right)\right)
            \right.
\nonumber\\
&\times\left.\frac{\prod_{m=1}^{d'}\left(\Phi\left(\tilde{a}^+_m\left(\mathbf{U}^{(m-1)}\right)\right)
-
\Phi\left(\tilde{a}^-_m\left(\mathbf{U}^{(m-1)}\right)\right)\right)}{\mathbb{P}\left(a_i^-\le\mathbf{Z}\cdot\mathbf{e}_i\le
a_i^+,\,i=1,\dots,d'\right)} \right]
\end{align}
\noindent where
\begin{equation}
\tilde{a}^{\pm}_i\left(\mathbf{U}^{(i-1)}\right) =
\frac{a_i^{\pm}-\sum_{j=1}^{i-1}\mathbf{e}_i\cdot\mathbf{f}_j
\Phi^{-1}\left(
\Phi\left(\tilde{a}_j^-\left(\mathbf{U}^{(j-1)}\right)
                \right) +
U_j\left(\Phi\left(\tilde{a}_j^+\left(\mathbf{U}^{(j-1)}\right) -
\Phi\left(\tilde{a}^-_j\left(\mathbf{U}^{(j-1)}\right)\right)\right)
\right)\right)}{\|f'_i\|}
\end{equation}
with the notation, $\mathbf{U}^{(j)}=\left(U_1,\dots,U_j\right),
j=1\dots,n$ and $U_1,\dots,U_{d'}$ i.i.d. uniformly distributed
random variables, all independent on $\mathbf{Z}$; we assume $U_0=0$
and $\tilde{a}^{\pm}_1= a^{\pm}_1$.
\end{prop}
\begin{rem}
The above result requires the computation of the joint probability
$\mathbb{P}\left(a_i^-\le\mathbf{Z}\cdot\mathbf{e}_i\le
a_i^+,\,i=1,\dots,d'\right)$ where the random variables
$\mathbf{Z}\cdot\mathbf{e}_i,\,i=1,\dots,d'$ are not independent; in
contrast, this term is not necessary for the estimation of
$\mathbb{E}\left[g(\mathbf{Z})\right]$. Indeed, suppose $K$ strata,
by conditioning we have:
\begin{equation}
\mathbb{E}\left[g(\mathbf{Z})\right]=
\sum_{k=1}^K\mathbb{E}\left[g(\mathbf{Z})\left|\mathbf{Z}\in
\text{$k$-th stratum}\right.\right]\mathbf{P}\left(\mathbf{Z}\in
\text{$k$-th stratum}\right),
\end{equation}
then plugging in the conditional expectation the result of equation
(\ref{nonorth:eq}) the probabilities at the numerator and at the
denominator simplify out.
\end{rem}
\begin{proof}
For simplicity we suppose $d'=2$, the Gram-Schmidt procedure returns
$\mathbf{f}'_1 = \mathbf{f}_1 = \mathbf{e}_1$,
$\mathbf{f}'_2=\mathbf{e}_2-(\mathbf{e}_1\cdot\mathbf{e}_2)\mathbf{e}_1$
and $\mathbf{f}_2 = \frac{\mathbf{f}'_2}{\|\mathbf{f}'_2\|}$. It
follows that
\begin{equation*}
\mathbb{E}\left[g(\mathbf{Z}) \left|
\begin{array}{c}
a_1^-\le\mathbf{Z}\cdot\mathbf{e}_1\le a_1^+\\
a_2^-\le\mathbf{Z}\cdot\mathbf{e}_2\le a_2^+
\end{array}
\right. \right]=\mathbb{E}\left[g(\mathbf{Z}) \left|
\begin{array}{c}
a_1^-\le\mathbf{Z}\cdot\mathbf{f}_1\le a_1^+\\
\frac{a_2^- -(\mathbf{e}_1\cdot\mathbf{e}_2)\mathbf{e}_1\cdot
\mathbf{Z}}{\|f'_2\|} \le\mathbf{Z}\cdot\mathbf{f}_2\le \frac{a_2^+
- (\mathbf{e}_1\cdot\mathbf{e}_2)\mathbf{e}_1\cdot
\mathbf{Z}}{\|f'_2\|}
\end{array}
\right. \right].
\end{equation*}
Based on the results of the Section \ref{sect:SLP} and the
properties of the conditional expectation, the previous expression
equals:
\begin{equation}
C\mathbb{E}\left[g\left(\mathbf{Z} +
\mathbf{f}_1\left(\Phi^{-1}\left(\Phi(a_1^-)+U_1\left(\Phi(a_1^+)-\Phi(a_1^-)\right)
\right)-\mathbf{f}_1\cdot\mathbf{Z}\right)\right)\indic_{\tilde{a}^-_2(U_1)\le\mathbf{f}_2\cdot\mathbf{Z}\le\tilde{a}^+_2(U_1)}\right]
\end{equation}
where $$ C = \frac{\Phi(a_1^+) -
\Phi(a_1^-)}{\mathbb{P}\left(\begin{array}{c}
a_1^-\le\mathbf{Z}\cdot\mathbf{e}_1\le a_1^+\\
a_2^-\le\mathbf{Z}\cdot\mathbf{e}_2\le a_2^+ \end{array}\right)}.
$$
The expected value is then:
\begin{align*}
&\int_0^1\mathbb{E}\left[g\left(\mathbf{Z} +
\mathbf{f}_1\left(\Phi^{-1}\left(\Phi(a_1^-)+u_1\left(\Phi(a_1^+)-\Phi(a_1^-)\right)
\right)-\mathbf{f}_1\cdot\mathbf{Z}\right)\right)\indic_{\tilde{a}^-_2(u_1)\le\mathbf{f}_2\cdot\mathbf{Z}\le\tilde{a}^+_2(u_1)}\right]du_1\\
&\qquad=\int_0^1\mathbb{E}\left[g\Bigg(\mathbf{Z} +
\mathbf{f}_1\left(\Phi^{-1}\left(\Phi(a_1^-)+u_1\left(\Phi(a_1^+)-\Phi(a_1^-)\right)
\right)-\mathbf{f}_1\cdot\mathbf{Z}\right)\right.\\
&\qquad+\left.\mathbf{f}_2\left(\Phi^{-1}\left(\Phi(\tilde{a}^-_2(u_1))+U_2\left(\Phi(\tilde{a}^+_2(u_1))-
\Phi(\tilde{a}^-_2(u_1))\right)
\right)-\mathbf{f}_2\cdot\mathbf{Z}\right)\Bigg)\right.\\
&\qquad\times\left.\left(\Phi(\tilde{a}^+_2(u_1))-\Phi(\tilde{a}^-_2(u_1))\right)\right]du_1\\
&\qquad=\mathbb{E}\left[g\Bigg(\mathbf{Z} +
\mathbf{f}_1\left(\Phi^{-1}\left(\Phi(a_1^-)+U_1\left(\Phi(a_1^+)-\Phi(a_1^-)\right)
\right)-\mathbf{f}_1\cdot\mathbf{Z}\right)\right.\\
&\qquad+\left.\mathbf{f}_2\left(\Phi^{-1}\left(\Phi(\tilde{a}^-_2(U_1))+U_2\left(\Phi(\tilde{a}^+_2(U_1))-
\Phi(\tilde{a}^-_2(U_1))\right)
\right)-\mathbf{f}_2\cdot\mathbf{Z}\right)\Bigg)\right.\\
&\qquad\times\left.\left(\Phi(\tilde{a}^+_2(U_1))-\Phi(\tilde{a}^-_2(u_1))\right)
\right].
\end{align*}
Rearranging the terms in $\mathbf{Z}$ we get equation
(\ref{nonorth:eq}) for $d'=2$. The result for $d'$ direction is
obtained iterating the steps above.
\end{proof}
\section{Convenient Directions}\label{sect:CD}
Given an allocation rule, the crucial point in the stratification of
linear projections is the choice of the directions of
stratification. Indeed, stratified sampling eliminates the sampling
variability across strata without affecting the sampling variability
within strata. Good directions are characterized by their higher
capacity to dissect the state space into strata where the integrand
function is nearly constant. In the following we describe the
approaches that we adopt in order to find the directions of
stratification.

\subsection{Principal Component Directions}
Suppose we want to find the singled-factor approximation of a
$d$-dimensional Gaussian random vector
$\mathbf{X}\sim\mathcal{N}(0,\Sigma)$ that maximizes the variance of
$\mathbf{v}\cdot\mathbf{X}$. This is equivalent to the following
optimization problem:
\begin{equation}
  \arg\max_{\|\mathbf{v}\|=1} \quad \mathbf{v}\cdot \Sigma \mathbf{v}
\end{equation}
\noindent Suppose $\lambda_1\ge\dots\ge\lambda_d$ represent the
eigenvalues  of $\Sigma$ in increasing order, and
$\mathbf{e}_1,\dots,\mathbf{e}_d$ their associated eigenvectors,
then the optimization above is solved by $\mathbf{v}^*=\mathbf{e}_1$
an eigenvector associated to the largest eigenvalue $\lambda_1$.

As $\mathbf{e}_1$ produces the linear combination
$\mathbf{e}_1\cdot\mathbf{X}$ that best captures the variability of
the components of $\mathbf{X}$. We may choose this vector as the
first direction of stratification. In the case we would consider
multiple stratification, we can iterate the optimization above. This
means that we would consider $\mathbf{e}_j, j=1,\dots,d$, associated
to the $j$-th eigenvalue, as the $j$-th direction of stratification.
Indeed, in the statistical literature, the linear combinations
$\mathbf{e}_j\cdot\mathbf{X}, j=1,\dots,d$, are called the principal
components of $\mathbf{X}$. The variance explained by the first
$k\le d$ principal components is the ratio:
\begin{equation*}
    \frac{\sum_{i=1}^k \lambda_i}{\sum_{i=1}^d \lambda_i}
\end{equation*}
Finally, we note that this procedure based on the PCA only produces
orthogonal directions.

\subsection{Law of Total Variance and GHS Directions}
In this section we illustrate the law of total variance and we
briefly describe the strategy to select optimal directions
illustrated in Glasserman et al. \cite{GHS99}. Given two random
vectors $\mathbf{X}_1$ and $\mathbf{X}_2$ of dimension  $d_1$ and
$d_2$, respectively, and a
 function $g:\mathbb{R}^{d_1}\rightarrow\mathbb{R}$, if
 $\mathbb{E}[g(\mathbf{X})^2]<\infty$, the law of total variance reads as:
 \begin{equation}\label{eq:totVar}
 \Var\left[g(\mathbf{X}_1)\right] = \mathbb{E}\left[\Var[g(\mathbf{X}_1)|\mathbf{X}_2]\right] +
 \Var[\mathbb{E}[g(\mathbf{X}_1)|\mathbf{X}_2]].
 \end{equation}
Usually, in the context of linear model, the two terms  are known as
the ``unexplained'' and the ``explained'' components of the
variance, respectively. In our case, $\mathbf{X}_1$ is a standard
normal random vector $\mathbf{Z}$ and
$\mathbf{X}_2=\mathbf{v}\cdot\mathbf{Z}$ where
$\mathbf{v}\in\mathbb{R}^{d}$. It is well known that stratification
eliminates the ``explained'' component of the variance up to terms
with order $o(1/N_S)$, where $N_S$ is the total number of draws (see
for instance Lemma 4.1 in Glasserman et al. \cite{GHS99}). Hence, a
good direction candidate is the one that maximizes the ``explained''
component of the variance or minimizes the ``unexplained'' part.

Such an optimal direction is then the solution of the following
optimization problem:
\begin{equation}\label{eq:optNew}
\mathbf{v}^*=\arg\min_{\mathbf{v}\in\mathbb{R}^d,\|\mathbf{v}\|=1}\int_{\mathbb{R}^d}\Var\left[g(\mathbf{Z})\Big|\mathbf{v}\cdot\mathbf{Z}=x\right]p_X(x)dx,
\end{equation}
where $p_X$ is the density of $X=\mathbf{v}\cdot\mathbf{Z}$.

The approach proposed in  Glasserman et al. \cite{GHS99} is to adopt
directions that are optimal for the quadratic approximation of the
logarithm of the integrand function. Glasserman et al. \cite{GHS99}
considered $g(\mathbf{z})=\exp{\left(\frac{1}{2}\mathbf{z}\cdot
B\mathbf{z}\right)}$ with $B$ non-singular symmetric matrix whose
eigenvalues $\lambda_1,\dots,\lambda_d$ are all less than $1/2$. Now
number the eigenvalues and eigenvectors of the matrix $B$ so that
\begin{equation}
\left(\frac{\lambda_1}{1-\lambda_1}\right)^2\ge\left(\frac{\lambda_2}{1-\lambda_2}\right)^2\ge\left(\frac{\lambda_d}{1-\lambda_d}\right)^2.
\end{equation}
Glasserman et al. \cite{GHS99} proved that the optimal direction
$\mathbf{v}^*$ is the eigenvector $\mathbf{e}_1$ of the matrix $B$
associated with the eigenvalue $\lambda_1$. When one considers
multiple stratification, the $j$-th optimal direction is the
eigenvector $\mathbf{e}_j$ associated with the eigenvalue
$\lambda_j$. Since the directions are the eigenvectors of the matrix
$B$, the GHS approach only produces orthogonal directions.

When the logarithm of the integrand function is not quadratic, one
could evaluate its Hessian at the certain point. Glasserman et al.
\cite{GHS99} proposed to calculate the Hessian at a point used for
an importance sampling procedure. This last operation might be
really computationally expansive, in particular if $d$ is large. It
depends on a non-convex optimization procedure and cannot always be
easily applied to realistic situations arising in finance. In
addition, in financial applications, payoff functions (integrand
functions) are far to be quadratic. In contrast, Etoré et al.
\cite{EFJM09} found the directions by adaptive techniques that in
some cases outperform the above approach. However, the numerical
procedure still remains computationally intensive. These drawbacks
motivate our study where our main purpose is to investigate
convenient multiple stratification directions that provide
comparable variance reductions with a notable advantage from the
computational point of view.


\subsection{Linear Approximations}\label{sec:LA}
In this section we describe a different approach, that we name
Linear Approximation (LA), in order to find convenient directions
for the stratification of linear projections.

Suppose $g\in\mathcal{C}^1$, this approach is based on a linear
approximation of the function $g$ that leads to an approximation of
the ``unexplained'' component of the variance. Then, we can
approximate the optimization problem (\ref{eq:optNew}) as:
\begin{eqnarray}\label{eq:Prop1}
\int_{\mathbb{R}^n}\nabla g(\mathbf{0})\cdot
\Var\left[\mathbf{Z}\Big|\mathbf{Z}\cdot\mathbf{v}=x\right]\nabla
g(\mathbf{0})p_X(x)dx,
\end{eqnarray}
where we also use the approximation $\nabla
g(\mathbb{E}[\mathbf{Z}\Big|\mathbf{Z}\cdot\mathbf{v}=x])\approx\nabla
g(\mathbb{E}[\mathbf{Z}])$, that is we evaluate the gradient at the
expected value of $\mathbf{Z}$ (zero for each component) instead of
its conditional one.
%
The solution of the optimization problem (\ref{eq:Prop1}) is given
by the following proposition:
\begin{prop}\label{prop:LA}
The optimal direction $\mathbf{v}^*$ of the optimization problem
(\ref{eq:Prop1}) is:
\begin{equation}
\mathbf{v}* = \pm\frac{\nabla g(\mathbf{0})}{\|\nabla
g(\mathbf{0})\|}
\end{equation}
\end{prop}
\begin{proof}
Developing equation (\ref{eq:Prop1}) we get:
\begin{eqnarray}
\int_{\mathbb{R}^d}\nabla g(\mathbf{0})\cdot
\Var\left[\mathbf{Z}\Big|X=x\right]\nabla g(\mathbf{0})p_X(x)dx&=&
\int_{\mathbb{R}^d}\nabla g(\mathbf{0})\cdot
(I-\mathbf{v}^T\mathbf{v})\nabla g(\mathbf{0})p_X(x)dx=
\nonumber\\
&&\|\nabla g(\mathbf{0})\|^2-\nabla
g(\mathbf{0})\cdot\mathbf{v}^T\mathbf{v}\nabla g(\mathbf{0}).
\end{eqnarray}
The minimization problem is equivalent to maximize the second term
that can be written as $\left(\nabla
g(\mathbf{0})\cdot\mathbf{v}\right)^2$. The maximum of this dot
product is attained when the two vectors are parallel. The optimal
direction is then obtained by normalization.
\end{proof}
Multiple directions in the LA procedure can be produced calculating
the gradient at different points. For example, we might iteratively
consider $\mathbf{Z}_2=\nabla g\left(\nabla
g(\mathbf{0})\right),\dots,\mathbf{Z}_{d'}=\nabla g\left(\nabla
g(\mathbf{Z}_{d'-1})\right)$ in order to capture higher order
components. We remark that the LA approach does provide
non-orthogonal directions.

\subsection{Linear Transformations}
The LT procedure, proposed by Imai and Tan \cite{IT2006}, is
originally conceived to enhance the accuracy of simulation
techniques that employ low-discrepancy sequences also known as
Quasi-Monte Carlo (QMC) methods. Indeed, given
$\mathbf{Z}\sim\mathcal{N}(0,I_d)$, the variance of the MC
estimation of the expected value $\mathbb{E}[g(\mathbf{Z})]$ does
not change if we replace $\mathbf{Z}$ by $A \epsilon$ where
$\epsilon\sim\mathcal{N}(0,I_d)$ and $A$ is a $d\times d$ orthogonal
matrix, $AA^T=I_d$, while the choice of $A$ can deeply affect the
accuracy of QMC simulations (see for instance Papageorgiou
\cite{Pap02}). The Imai and Tan's choice is  such that $A$ minimizes
the effective dimension in the truncation sense defined in Caflisch
et al. \cite{CMO1997} of the integrand function. In our context, the
columns of $A$ will be chosen as the orthogonal directions of
stratification.

We briefly describe the LT algorithm. Consider a $d$ dimensional
normal random vector $\mathbf{T}\sim \mathcal{N}(\mu;\Sigma)$, a
vector $\mathbf{w}=(w_1,\dots,w_d)\in\mathbb{R}^d$ and let
$f(\mathbf{T}) =\sum_{i=1}^d w_iT_i$ be a linear combination of
$\mathbf{T}$. Let $C$ be such that $\Sigma=CC^T$ and assume
$\epsilon\sim\mathcal{N}(0,I_d)$ with
$\mathbf{T}\stackrel{\mathcal{L}}{=}C\epsilon$. The LT approach
considers $C$ as $C=C^{\text{LT}}=C^{\text{CH}}A$, with
$C^{\text{CH}}$ the Cholesky decomposition of $\Sigma$. Then, in the
linear case, we can define:
\begin{equation}\label{4.3.4}
    g^{A}(\epsilon):= f(C^{\text{CH}}A\epsilon) = \sum_{k=1}^d \alpha_k \epsilon_k + \mu\cdot
    \mathbf{w},
\end{equation}
\noindent where $\alpha_k= \mathbf{C^{\text{LT}}_{\cdot k}}\cdot
\mathbf{w}=\mathbf{A_{\cdot k}}\cdot\mathbf{B},\,k=1\dots,d$ and
$\mathbf{B}=(C^{\text{CH}})^T\mathbf{w}$ while $\mathbf{C_{\cdot
k}}$ and $\mathbf{A_{\cdot k}}$ are the $k$-th columns of the matrix
$C$ and $A$, respectively. In the linear case, setting
\begin{equation}\label{4.3.8}
    \mathbf{A_{\cdot 1}^*} = \pm\frac{\mathbf{B}} {\|\mathbf{B}\|},
\end{equation}
with arbitrary remaining columns with the only constrain that
$AA^T=I_d$, leads to the following expression:
\begin{equation}
g^A(\epsilon)=\mu\cdot\mathbf{w}\pm\|\mathbf{B}\|\epsilon_1.
\end{equation}
This is equivalent to reduce the effective dimension in the
truncation sense to $1$ and this means to maximize the variance of
the first component $\epsilon_1$.

In a non-linear framework, we can use the LT construction, which
relies on the first order Taylor expansion of $g^A$:
\begin{equation}\label{4.3.12}
    g^A(\epsilon) \approx g^A(\hat{\epsilon}) +
    \sum_{l=1}^d\frac{\partial
    g^A(\hat{\epsilon})}{\partial\epsilon_l}\Delta\epsilon_l.
\end{equation}
\noindent The approximated function is linear in the standard normal
random vector $\Delta\epsilon\sim\mathcal{N}(0,I_d)$ and we can rely
on the considerations above. The first column of the matrix $A^*$ is
then:
\begin{equation}
\mathbf{A_{\cdot 1}}^* =   \arg\max_{\mathbf{A_{\cdot 1}}\in
\mathbf{R^d}}\left(\frac{\partial
    g^A(\hat{\epsilon})}
    {\partial\epsilon_1}\right)^2
\end{equation}
Since we have already maximized the variance contribution for
$\left(\frac{\partial
g^A(\hat{\epsilon})}{\partial\epsilon_1}\right)^2$, in order to
improve the method using adequate columns  we might consider the
expansion of $g$ about $d-1$ different points. More precisely Imai
and Tan \cite{IT2006} propose to maximize:
\begin{equation}\label{4.3.13}
\mathbf{A_{\cdot k}}^* =   \arg\max_{\mathbf{A_{\cdot k}}\in
\mathbf{R^d}}\left(\frac{\partial
    g^A(\hat{\epsilon}_k)}
    {\partial\epsilon_k}\right)^2
\end{equation}
subject to $\|\mathbf{A_{\cdot k}}^*\|=1$ and $\mathbf{A_{\cdot
j}}^*\cdot \mathbf{A_ {\cdot k}}^*=0, j=1,\dots,k-1, k\le d$.

Although equation (\ref{4.3.8}) provides an easy solution at each
step, the correct procedure requires that the column vector
$\mathbf{A_{\cdot k}}^*$ is orthogonal to all the previous (and
future) columns. Imai and Tan \cite{IT2006} propose to choose
$\hat{\epsilon}=\hat{\epsilon}_1=\mathbb{E}[\epsilon]=\mathbf{0}$,
$\hat{\epsilon}_2=(1,0,\dots,0),\dots\hat{\epsilon}_{k}=(1,1,1,\dots,0,\dots,0)$,
where the $k$-th point has $k-1$ leading ones.  Sabino \cite{Sab08}
illustrated an economic and convenient implementation of the LT
algorithm by an iterative QR decomposition that we will use to find
the directions of stratification. This method is computationally
more expensive than the LA and it is not clear if it admits a
solution when the sequence of expansion points is different from the
one described above.

\section{Financial Applications}\label{sect:FA}
In this section we illustrate how to calculate the convenient
directions introduced above in the context of option pricing. We
consider two price-dynamics:
\begin{itemize}
\item BS dynamics for $M$ risky assets with constant volatilities:
\begin{equation}
dS_{i}\left( t\right) =rS_{i}\left( t\right) dt+\sigma
_{i}S_{i}\left( t\right)\,dW_{i}\left( t\right) ,\quad
S_{i}(0)=S_{i0},\qquad i=1,\dots ,M,
\end{equation}
$S_{i}\left( t\right) $ denotes the $i$-th asset price at  time $t$,
$\sigma _{i}$ represents the volatility of the $i$-th asset return,
$r$ is the risk-free rate, and $\mathbf{W}\left( t\right) =\left(
W_{1}\left( t\right) ,\dots ,W_{M}\left( t\right) \right) $ is a
$M$-dimensional Brownian motion such that $dW_{i}(t)dW_{k}(t)=\rho
_{ik}dt,\, i,k = 1,\dots ,M$. When $M=1$ we simply denote
$S(t)=S_1(t)$.
\item CIR dynamics:
\begin{equation}\label{CIRDyn}
dS(t) = \alpha\left(\mu-S(t)\right)dt+\sigma\sqrt{S(t)}dW(t),\quad
S(0)=S_0,
\end{equation}
\noindent with $S_0, \alpha,\mu,\sigma$ positive constants. We
impose the condition $2\alpha\mu>\sigma^2$ in order to ensure that
$S(t)$ remains positive.
\end{itemize}
Applying the risk-neutral pricing formula (see Lamberton and Lapeyre
\cite{LL96}), the calculation of the price at  time $t$ of any
European derivative contract with maturity date $T$ boils down to
the evaluation of an (discounted) expectation:
\begin{equation}
a(t) =
\exp\left(-r(T-t)\right)\mathbb{E}\left[\psi\right|\mathcal{F}_t]\label{1.3}\text{,}
\end{equation}
\noindent  the expectation is under the risk-neutral probability
measure and $\psi$ is a generic $\mathcal{F}_T$-measurable variable
that determines the payoff of the contract.

We show how to derive the convenient directions of stratification
for the following derivative contracts:
\begin{enumerate}
\item discretely monitored Asian basket options:
\begin{equation}\label{eq:AsianPayoff}
a\left( t\right) =\exp\left(-r(T-t)\right)\mathbb{E}\left[\left(
\sum_{i=1}^{M}\sum_{j=1}^{N}w_{ij}\,S_{i}\left( t_{j}\right)
-K_S\right)^+\bigg|\mathcal{F}_t\right]\quad\text{Option on a
Basket}
\end{equation}%
\noindent where $x^+ =max(x,0)$, $t_1<t_2\dots<t_N=T$ is a time
grid, the coefficients $w_{ij}$ satisfy $\sum_{i,j}w_{ij}=1$ and
$K_S$ is the strike price. When $N=1$ and $M>0$ the option is known
as basket option while if $M=1$ and $N>0$ it is simply known as
Asian option.
\item Asian option with knock-out barrier at expiry $T$:
\begin{equation}\label{eq:AsianBarrierExpiry}
a\left( t\right) =\exp\left(-r(T-t)\right)\mathbb{E}\left[\left(
\frac{1}{N}\sum_{j=1}^{N}\,S\left( t_{j}\right)
-K_S\right)^+\indic_{S(T)<B}\bigg|\mathcal{F}_t\right]
\end{equation}
 where $B$ represents the value of the barrier.
\item Asian option with knock-out barrier at each monitoring time:
\begin{equation}\label{eq:AsianBarrier}
a\left( t\right) =\exp\left(-r(T-t)\right)\mathbb{E}\left[\left(
\frac{1}{N}\sum_{j=1}^{N}\,S\left( t_{j}\right)
-K_S\right)^+\indic_{S\left( t_{j}\right)<B\,,\forall
j=1,\dots,N}\bigg|\mathcal{F}_t\right]
\end{equation}
 where $B$ represents the value of the barrier.
\end{enumerate}
\subsection{Linear Transformation in the Black-Scholes Market}
Suppose the BS dynamics with constant volatilities and a time grid
$t_1<t_2\dots<t_N=T$, the elements of the autocorrelation matrix
$\Sigma_B$ of the Brownian motion are
$(\Sigma_B)_{jn}=\min(t_j,t_n),\,j,n=1,\dots,N$. Moreover, denote
$\Sigma_A$ the a covariance matrix whose elements are
$(\Sigma_A)_{im}=\sigma_i\rho_{im}\sigma_m$, $i,m=1,\dots,M$, and
consider $\Sigma_{MN}=\Sigma_B\otimes\Sigma_A$ where $\otimes$
denotes the Kronecker product. Given $\epsilon\sim{N}(0,I_{MN})$ and
$C^{\text{LT}}=C^{\text{CH}}A$ such that
$C^{\text{CH}}(C^{\text{CH}})^T=\Sigma_{MN}$ and $AA^T=I_{MN}$, the
payoff of an Asian basket option can written as:
\begin{equation}
\psi = \left(g(\epsilon)-K_S\right)^+\label{eq:LT1}\quad
\text{where}\quad
 g(\epsilon)=\sum_{k=1}^{MN}\exp\left\{\mu_k + \sum_{l=1}^{MN}C^{\text{LT}}_{kl}\epsilon_l \right\}
\end{equation}
and
\begin{equation}
\mu_k = \ln(w_{k_1k_2}S_{k_1}(0)) + \bigg(r-\frac{\sigma_{k_1}^2}{2}
\bigg)t_{k_2}\label{eq:LT3}
\end{equation}
where the indexes $k_1$ and $k_2$ are $k_1=(k-1)\text{modulo}M + 1,
k_2 = \lfloor(k-1)/M\rfloor+1$, respectively and $\lfloor x\rfloor$
denotes the greatest integer less than or equal to $x$.

Since the Asian payoff function is not  everywhere differentiable,
the LT procedure is applied to its differentiable part $g$ (or
$g-K_S$). This is done also for the other barrier-style Asian
options, hence we obtain the same directions of stratification  for
the three types of derivative contracts. Hereafter we detail the
adopted procedure:
\begin{enumerate}
\item
Expand $g$ up to the first order:
\begin{equation}\label{eq:LT5}
    g(\epsilon) \cong g(\hat{\epsilon}) +
    \sum_{l=1}^{NM}\left(\sum_{i=1}^{NM}\exp\left(\mu_i+\sum_{k=1}^{NM}C^{\text{LT}}_{ik}\hat{\epsilon}_k\right)C^{\text{LT}}_{il}\right)\Delta\epsilon_l
\end{equation}
\item
\noindent For $\hat{\epsilon}=\mathbf{0}$ find the first column of
the optimal matrix $A$:
\begin{equation}\label{eq:LT6}
    g(\epsilon) \cong g(\mathbf{0}) +
    \sum_{l=1}^{NM}\left(\sum_{i=1}^{NM}\exp\left(\mu_i\right)C^{\text{LT}}_{il}\right)\Delta\epsilon_l
\end{equation}
\noindent Set
$\alpha_l=\left(\sum_{i=1}^{NM}\exp\left(\mu_i\right)C^{\text{LT}}_{il}\right)=
\sum_{m=1}^{NM}\left(\sum_{i=1}^{NM}\exp\left(\mu_i\right)C_{im}^{\text{CH}}\right)A_{ml}$
and set $\mathbf{u^{(1)}} = (e^{\mu_1},\dots,e^{\mu_{MN}})^T$ and
$\mathbf{B^{(1)}}= (C^{\text{CH}})^T\mathbf{u^{(1)}}$ then the first
column is
\begin{equation}\label{eq:LT7}
\mathbf{A_{\cdot
1}^*}=\pm\frac{\mathbf{B^{(1)}}}{\|\mathbf{B^{(1)}}\|}.
\end{equation}
\item The $p$-th optimal column is found considering
the $p$-th expansion point of the strategy. This results in:
\begin{equation}
    g(\epsilon) \cong g(\hat{\epsilon}_p) +
    \sum_{l=1}^{NM}\left(\sum_{i=1}^{NM}\exp\left(\mu_i+\sum_{k=1}^{p-1}C_{ik}^*\right)C^{\text{LT}}_{il}\right)\Delta\epsilon_l
\end{equation}
\noindent where $C_{ik}^*=(C^{\text{CH}}A^*_k)_i$, $k<p$ have been
already found at the $p-1$ previous steps and $\mathbf{A_{\cdot
p}}^*$ must be orthogonal to all the other columns.

Also define $\mathbf{u^{(p)}} =
\left(\exp\left(\mu_1+\sum_{k=1}^{p-1}C_{1k}^*\right),\dots,\exp\left(\mu_{MN}+\sum_{k=1}^{p-1}C_{MNk}^*\right)\right)^T$
and $\mathbf{B^{(p)}}= (C^{\text{CH}})^T\mathbf{u^{(p)}}$, then the
solution is
\begin{equation}\label{eq:LT8}
\mathbf{A_{\cdot
p}^*}=\pm\frac{\mathbf{B^{(p)}}}{\|\mathbf{B^{(p)}}\|}.
\end{equation}
We remark that at each time step all the columns must be
orthogonalized (see Sabino \cite{Sabino08b,Sab08})
\end{enumerate}
\subsection{Linear Transformation in the CIR Market}
We extend the procedure described in the previous section with the
assumption of a CIR dynamics. Consider an equally spaced time-grid
whose time step is denoted by $\Delta t$, the Euler scheme of the
CIR dynamic is:
\begin{equation}\label{eq:EulCIR}
S_j = S_{j-1} + \alpha\left(\mu-S_{j-1}\right)\Delta t +
\sigma\sqrt{S_{j-1}\Delta t}\,Z_j,\quad j=1,\dots N,
\end{equation}
where $\mathbf{Z}$ is a Gaussian vector of $N$ independent standard
random variables. The Asian payoff is:
\begin{equation}
\psi = \left(h(\mathbf{Z})-K_S\right)^+\quad\text{with}\quad
h(\mathbf{Z}) =
\frac{1}{N}\sum_{j=1}^NS_j(\mathbf{Z}).\label{eq:CIRAsian}
\end{equation}
As done in the BS setting, we find the LT-based convenient
directions of stratification applying the LT technique to the
differentiable part of the payoff function of an Asian option (in
this dynamics we only consider options on a single asset). This is
done also for the other barrier-style Asian options, so that we have
the same directions of stratification  for the three types
derivative contracts. Applying the LT decomposition the Euler scheme
becomes
\begin{equation}\label{eq:EulCIR:LT}
S_j = S_{j-1} + \alpha\left(\mu-S_{j-1}\right)\Delta t +
\sigma\sqrt{S_{j-1}\Delta t}\sum_{m=1}^NA_{jm}\epsilon_m,\quad
j=1,\dots N,
\end{equation}
the computation of the first direction of LT decomposition consists
in the following steps:
\begin{enumerate}
\item Compute the partial derivatives $\frac{\partial
S_j}{\partial\epsilon_1}$, $j=1,\dots,N$:
\begin{equation}\label{eq:EulDer}
\frac{\partial S_j(\mathbf{0})}{\partial\epsilon_1}=
\left\{\left[1-\alpha\Delta t +\frac{\sigma}{2}\sqrt{\frac{\Delta
t}{S_{j-1}}}\sum_{m=1}^NA_{jm}\epsilon_m\right]\frac{\partial
S_{j-1}}{\partial\epsilon_1}+\sigma\sqrt{\Delta
tS_{j-1}}A_{j1}\right\}\Big|_{\epsilon=\mathbf{0}}.
\end{equation}
Now denote $p_j^{(1)} = \frac{\partial
S_j(\mathbf{0})}{\partial\epsilon_1}$,
$\alpha_{j-1}^{(1)}=\left(1-\alpha\Delta t
+\frac{\sigma}{2}\sqrt{\frac{\Delta
t}{S_{j-1}}}\sum_{m=1}^NA_{jm}\epsilon_m \right)\Big|_{\epsilon=
\mathbf{0}}$ and $\beta_{j-1}^{(1)}= \sigma\sqrt{\Delta
tS_{j-1}(\mathbf{0})}$, we have
\begin{equation}\label{eq:iterLT1}
p_j^{(1)} = p_{j-1}^{(1)} \alpha_{j-1}^{(1)} +
\beta_{j-1}^{(1)}A_{j1}.
\end{equation}
\begin{rem}
The third term in $\alpha^{(1)}$ is zero, nevertheless we show its
expression because the results below still hold when we compute the
vector $\alpha^{(l)}$ of parameters in the $l$-th step, where we
consider $\epsilon_l = (\underbrace{1,1,\dots,1}_{l-1 \text{
times}},0,\dots,0)$, $l=1,\dots,N$.
\end{rem}
\begin{prop}\label{prop:LT}
The solution of the recurrence equation (\ref{eq:iterLT1}) is a
linear combination of the rows of $A$:
\begin{equation}\label{eq.sol1}
p_j^{(1)} = \sum_{m=1}^jw_m^{(1)}(j)A_{m1}\quad, j=1,\dots,N,
\end{equation}
where the components of vector $\mathbf{w}^{(1)}(j)$, that depends
on $j$, are:
\begin{equation}
w_m^{(1)}(j) = \beta_{m-1}^{(1)}\prod_{i=m}^{j-1}\alpha_i^{(1)}.
\end{equation}
\end{prop}
The superscripts indicate the number of the direction under
consideration and the proof can be obtained by iteration.
\begin{rem}\label{rem:1}
Note that $w_j^{(1)}(j) = \beta_{j-1}^{(1)}$ with the assumption
that $\prod_{i\in\emptyset}\alpha_i^{(1)}=1$ and
$w_m^{(1)}(j+1)=\alpha_j^{(1)}w^{(1)}_m(j),\,\forall j,m$.
\end{rem}
\item Denote $\tilde{h}(\epsilon)=h(\mathbf{Z})= h(A\epsilon)$ then
\begin{equation}\label{eq:CIROpt} \frac{\partial
\tilde{h}(\mathbf{0})}{\partial\epsilon_1}=\frac{1}{N}\sum_{j=1}^Np_j^{(1)}.
\end{equation}
\begin{corr}\label{corr1}
$\frac{\partial\tilde{h}}{\partial\epsilon_l}\Big|_{\epsilon_1=\mathbf{0}}$
in equation (\ref{eq:CIROpt}) is a linear combination of the rows of
$A$:
\begin{equation}\label{eq:sol2}
\sum_{j=1}^Np_j^{(1)} = \sum_{j=1}^Nt_j^{(1)}A_{j1},\quad\forall
N\in\mathbb{N},
\end{equation}
where
\begin{equation}
t_j^{(1)} =
\beta_{j-1}^{(1)}\left(1+\sum_{l=j}^{N-1}\prod_{i=j}^l\alpha_i^{(1)}\right).
\end{equation}
\end{corr}
As for Proposition \ref{prop:LT}, the proof can be obtained by
iteration.
\begin{rem}\label{rem:2}
$t_N^{(1)}=\beta^{(1)}_{N-1}=w_N^{(1)}(N)$.
\end{rem}
\item The first optimal direction is established by the following
theorem.
\begin{theo}\label{th:CIR}
The first column of the matrix $A$, solution of the LT optimization
problem, in the case of Asian options assuming the Euler
discretization of the CIR model is:
\begin{equation}\label{eq:CIRSol}
\mathbf{A_{\cdot l}}^* =
\frac{\mathbf{t}^{(1)}}{\|\mathbf{t}^{(1)}\|},
\end{equation}
with $t$ being the vector defined in Corollary \ref{corr1}.
\end{theo}
\begin{proof}
Knowing that the scalar product
$\mathbf{t}^{(1)}\cdot\mathbf{A_{\cdot 1}}$ attains the maximum when
the two vectors are parallel, we can conclude that the optimal
$\mathbf{A}_{\cdot 1}^*$ is proportional to $\mathbf{t}^{(1)}$.
After normalization the optimum solution is given by equation
(\ref{eq:CIRSol}).
\end{proof}
\begin{rem}
We observe that, if $\mathbf{Z}=\mathbf{0}$, after some algebra, the
Euler discretization is simply
\begin{equation}
S_j-\mu = \left(1-\alpha\Delta t\right)\left(S_{j-1} -\mu\right)
\end{equation}
then
\begin{equation}
S_j=\left(1-\alpha\Delta t\right)^j\left(S_{0} -\mu\right) + \mu
\end{equation}
We use the results of this remark to simplify the computational cost
to find the first direction of stratification.
\end{rem}
\item In order to compute the remaining optimal columns we need to
repeat the procedure illustrated in steps 1 to step 3. As far as the
calculation of the $l$-th column is concerned, one needs to evaluate
$\frac{\partial S_j(\hat{\epsilon}_l)}{\partial\epsilon_l}$ and
accordingly the quantities $p_j^{(l)}$, $\alpha_j^{(l)}$,
$\beta_j^{(l)},\, \forall j$, and the components of the vectors
$\mathbf{w}^{(l)}$ and $\mathbf{t}^{(l)}$. All the results in
Proposition \ref{prop:LT}, Corollary \ref{corr1} and Theorem
\ref{th:CIR} remain valid while now considering the quantities with
superscripts $l$. The orthogonal directions LT are then obtained by
orthogonalization.
\end{enumerate}
\subsection{Linear Approximation in the Black-Scholes Market}
Hereafter we describe how to find the directions of the LA technique
in the case of a BS dynamics. Since the payoff function is not
differentiable, as for the LT method we consider only the
differentiable part $g-K_S$. The gradient has components:
\begin{equation*}
 \frac{\partial g(\epsilon)}{\partial \epsilon_m}=\sum_{k=1}^{MN}C_{km} \exp\left\{\mu_k + \sum_{l=1}^{MN}C_{kl}\epsilon_l\right\},
\end{equation*}
then,
\begin{equation}
\nabla g(\mathbf{0})=\left[
\begin{array}{c}
 \sum_{k=1}^{MN}C_{k1}e^{\mu_1}\\
\vdots\\
 \sum_{k=1}^{MN}C_{kMN}e^{\mu_{MN}}\\
\end{array}
\right]\quad\text{and in general}\quad \nabla
g(\mathbf{\hat{\epsilon}})=\left[
\begin{array}{c}
 \sum_{k=1}^{MN}C_{k1}e^{\mu_1+\hat{\epsilon}_1}\\
\vdots\\
 \sum_{k=1}^{MN}C_{kMN}e^{\mu_{MN}+\hat{\epsilon}_{MN}}\\
\end{array}
\right].
\end{equation}
In the above derivation we assume that $C=C^{\text{CH}}$ since we do
not need to introduce any orthogonal matrix and the Cholesky
decomposition of the autocorrelation matrix of a Brownian motion is
 explicitly known. It turns out that the LT  and the LA methods
return the same first order  direction. Nevertheless, the latter
approach can produce different directions changing the value at
which the gradient is calculated. In contrast, the LT procedure
admits solution only assuming the starting points strategy described
above. Hence, the LA is more flexible and in particular the new
algorithm does not require an incremental QR decomposition to find
the new directions. Indeed, if we would look for orthogonal
directions  a unique orthogonalization would be required;
consequently, the LA computational cost is much lower. Moreover, the
mathematical derivation is simpler.
\subsection{Linear Approximation in the CIR Market}
We now illustrate how to apply the new LT approach for the
derivative contracts above in CIR dynamics.  Consider the Euler
discretization scheme in equation (\ref{eq:EulCIR}) and compute the
following partial derivatives for $j,l=1,\dots,N$:
\begin{equation*}
\frac{\partial S_j}{\partial Z_l}= \left[1-\alpha\Delta t
+\frac{\sigma}{2}\sqrt{\frac{\Delta
t}{S_{j-1}}}Z_j\right]\frac{\partial S_{j-1}}{\partial
Z_l}+\sigma\sqrt{\Delta tS_{j-1}}\delta_{jl},
\end{equation*}
then
\begin{equation}
\frac{\partial S_j(\mathbf{0})}{\partial Z_l}= \left(1-\alpha\Delta
t \right) \frac{\partial S_{j-1}(\mathbf{0})}{\partial
Z_l}+\sigma\sqrt{\Delta tS_{j-1}(\mathbf{0})}\delta_{jl},
\end{equation}
and the gradient is
\begin{equation}
\nabla S_j(\mathbf{0})=\left[
\begin{array}{c}
 \left(1-\alpha\Delta t \right)^{j-1}\sigma\sqrt{\Delta tS_0}\\
\left(1-\alpha\Delta t \right)^{j-2}\sigma\sqrt{\Delta tS_1(\mathbf{0})}\\
\vdots\\
\left(1-\alpha\Delta t \right)\sigma\sqrt{\Delta tS_{j-2}(\mathbf{0})}\\
\sigma\sqrt{\Delta tS_{j-1}(\mathbf{0})}\\
0\\
\vdots\\
0
\end{array}
\right].
\end{equation}
Due to Proposition \ref{prop:LA}, the LA first optimal direction  is
given by the normalized sum of $\nabla
S_j(\mathbf{0}),\,j=1,\dots,N$. Further directions are obtained by
iterating this procedure with a starting points rule. Alternatively,
we can choose the evaluation points  as in the LT strategy or the
components of the $l$-th direction for the starting point of the
gradient for the $l+1$-direction.
\section{Numerical Illustrations}\label{sect:NI}
We now illustrate the results developed in the previous sections
through examples and numerical experiments. As mentioned before, we
consider the BS and the CIR dynamics and different exotic
path-dependent options. All the numerical procedures have been
implemented in MATLAB on a computer with Intel Pentium M, 1.60 GHz,
1 GB RAM. In the numerical illustrations we consider $K=1000$ strata
and $N_S=2\times10^6$ total number of scenarios so that for
orthogonal directions we have a constant allocation rule (which, in
this case, coincides the proportional rule as the strata are
equiprobable) with $2000$ random draws in each stratum
(\textit{const} in the tables). When we consider non-orthogonal
directions the constant allocation rule is not proportional anymore
since the strata are not equiprobable. For the optimal allocation
rule (\textit{opt}), the standard deviations have been computed by a
first pilot run and then they have been used in a second stage to
determine the stratified estimator.

We report the estimated variances  and the total computational times
with constant and optimal allocation. We compare the variances
employing the directions of stratification returned by GHS (see
Glasserman et al. \cite{GHS99}), LT, LA, the PCA and their
combination. Note that the GHS procedure requires the calculation of
an importance sampling direction that is a computationally demanding
task. In our experiments we report the variances due to the
stratification only in order to compare the relative efficiency of
the pure stratification methods. As far as the PCA directions are
concerned, they consist of the eigenvectors associated to the
highest eigenvalues of the autocorrelation matrix of the
multi-dimensional Brownian motion that drives the BS dynamics. In
contrast, since the CIR dynamics is not Gaussian, in a first pilot
run with a $2000$-sample  we compute the MC estimation of the
autocorrelation matrix of the price dynamics and then calculate its
eigenvectors and values. We employ a Euler scheme that always takes
the positive value of the square-root term because it was shown that
this exhibits the smallest discretization bias among Euler
CIR-discretizations (see Andersen \cite{And2007}). Even if this
dynamics is not normal, the $i$-th step price, given the $i-1$-th
one, is normal and this can justify the use of the PCA in the CIR
dynamics. We consider the multiple combination of two directions of
stratification. Our algorithm and considerations are also applicable
to additional directions but, due to the so called \textit{curse of
dimensionality}, this would require a higher number of strata and
hence a higher number of total samples that would considerably
increase the computational burden. Finally, we compare these
stratified estimators to LHS-based estimators (see Owen
\cite{ow1992a} or Stein \cite{Stein87} for more on this topic).
Stein \cite{Stein87} proved that LHS eliminates the variance of the
additive part of the integrand (payoff) function and hence produces
an important variance reduction when coupled with LA or LT.
Unfortunately, it is difficult to numerically compute the asymptotic
variance in the central limit theorem for the LHS estimator. LHS is
characterized by a fixed multiple allocation rule that has a high
computational cost. Our purpose is to compare this very
high-dimensional allocation rule to one with a lower dimension where
we can adopt optimal allocation. In addition, the expectation of
interest $\mathbb{E}[\psi(\mathbf{Z})]$ is equal to
$\mathbb{E}[\psi(O\mathbf{Z})]$ where $O$ is a general orthogonal
matrix. In a standard MC simulation the variance of the two
estimators does not depend on $O$ but in contrast, the accuracy of
LHS-based estimators critically depend on the choice of $O$. Our
simulations adopt the orthogonal matrix produced by the LT
decomposition that has been shown to be an efficient choice (see
Sabino \cite{Sab08}).
\subsection{Asian Options in the Black-Scholes Market}
Our first example is the pricing of arithmetic Asian options on a
single risky security defined by equation (\ref{eq:AsianPayoff})
with $M=1$. For simplicity we assume that the time grid is regular
with time steps $t_i=i\Delta t, i=1,\dots,N$. This permits a simple
derivation of the PCA and the Cholesky decomposition of the
autocorrelation matrix of the Brownian motion (see {\AA}kesson and
Lehoczky \cite{AL1998}). Table \ref{tab:Asian} reports the input
parameters used in the simulation with  different moneyness of the
options. We remind that in this setting LT and LA provide the same
first order direction.
\begin{table}
\centering \caption{Input Parameters in the BS dynamics}
    \subtable[Arithmetic Asian Options]{\label{tab:Asian}
        \begin{tabular}{|c|c|c|c|c|c|c|c|}
        \hline
        $S_0$ & $K_S$ &  $N$ & $r$ & $\sigma$ & $T$ \\
        \hline
        $50$ & $45, 50, 55$ & $64$&$0.05$& $0.3$ & $1$ \\
        \hline
        \end{tabular}
    }
    \quad
    \subtable[Arithmetic Asian Barrier Options]{\label{tab:AsianBarrier}
            \begin{tabular}{|c|c|c|c|c|c|c|c|c|}
            \hline
            $S_0$ & $K_S$ &  $B$& $N$ & $r$ & $\sigma$ & $T$  \\
            \hline
            $50$ & $50, 55$ & $60, 70,80$ & $16$ & $0.05$ & $0.1$ & $1$ \\
            \hline
            \end{tabular}
    }
\end{table}
Tables ~\ref{tab:AsianResultsBS}-\ref{tab:BarrierCompleteResultsBS}
report the numerical results obtained and the total computational
times: all the procedures return unbiased estimates of the option
prices while giving remarkably different variances. All the
stratified techniques give a variance reduction that is particularly
remarkable with the GHS and the LA (LT)  methods. The PCA orthogonal
directions (one dimensional and two dimensional)  give a modest
effect also taking into account the computational times. The main
observation is that GHS and LA (LT) show the same computational cost
and the same variance reduction. Both LA and GHS give a remarkable
variance reduction, of a factor of more than $100$ in the case of
constant allocation and of several hundreds in the case of optimal
allocation. However, given the parameters in Table \ref{tab:Asian},
we stress the fact that the computational time required for the
calculation of the direction is really a small part of the total
time requested for all the proposed procedures. In contrast, with a
really high problem dimension (i.e. a dimension $1000$ typical in
financial applications), the solution of the GHS optimization
problem becomes a hard task depending on the starting guess and its
computational burden has a relevant influence. In contrast, the LA
(LT) algorithm consists in a simple vector $O(N)$ calculation that
is feasible even in high-dimensional problems. Table
\ref{tab:AnglesAsian} reports the angles (in degrees) between the
discussed directions. The GHS and LA directions are almost parallel
meaning that the GHS algorithm is not so sensitive to the moneyness
and this justifies the equal performance in terms of variance
reduction of the LA method. As mentioned before, the PCA direction
does not furnish a relevant variance direction and hence the
non-orthogonal $2$-dimensional  stratification that employs such a
direction always returns a lower accuracy than the GHS or LA
methods. Moreover, the orthogonal GHS or LA bi-dimensional
stratifications give variance reductions that are about $4$ times
lower than the corresponding one-dimensional ones. We remind that
the two settings have the same number of strata so that we can
conclude that the second order direction has a lower impact on the
variance reduction and, with these directions of stratification, it
is more efficient to employ a stratified MC estimator with a single
direction. We conclude the study for the simple Asian options with
the comparison between the accuracies of the LHS and the stratified
sampling with a single direction with optimal allocation. The
results shown in Table \ref{tab:AsianResultsBS} illustrate that  the
LHS never outperforms the optimal allocation. Indeed, the LHS-based
variance is at least two times the variance obtained with the
stratified estimator with optimal allocation. Moreover, the
computational cost is a lot higher, almost twice as high as the
times needed for the optimal allocation. All these arguments
strongly favor the use of convenient directions with optimal
allocation.

We modify the Asian option example by adding a knock-out barrier at
expiration $T$ or at each sampling date so that the option pays
nothing if the asset price is above the barrier. Due to the
discontinuous payoff of barrier options, the GHS optimization
problem is a demanding task especially when the barrier is at each
time step (indeed Glasserman et al. \cite{GHS99} did not elaborate
this possibility). In contrast, the LA (LT) focuses only  the
continuous part of the payoff function. Table \ref{tab:AsianBarrier}
reports the input parameters used in the simulation with different
moneyness and barriers. The values of the barriers should be larger
than the strike prices but not too high otherwise the pricing
problem would almost boil down into the case without barrier.
\begin{table}
    \centering
    \caption{Angles between the Stratifying Directions in
    degrees}
    \subtable[Arithmetic Asian Options]{\label{tab:AnglesAsian}
    \begin{tabular}{|c|c|c|c|}
            \hline
            &$K_S=45$ & $K_S=50$ & $K_S=55$\\
            \hline
            LA-GHS & $1.35$ &$1.04$& $1.74$ \\
            LA-PCA & $54.62$& $52.73$& $51.60$\\
            GHS-PCA & $56.60$&$53.83$ & $53.30$\\
            \hline
            \end{tabular}
    }
    \subtable[Arithmetic Asian Barrier Options]{\label{tab:AnglesBarrier}
        \begin{tabular}{|c|c|c|c|c|}
            \hline
            & \multicolumn{2}{|c|}{$K_S=50$} &
            \multicolumn{2}{|c|}{$K_S=55$} \\
            \hline
            &  $B=60$ & $B=70$ & $B=70$ & $B=80$\\
            \hline
                LA-GHS & $0.37$ & $0.37$ & $0.75$ & $0.75$\\
            LA-PCA & $51.95$ & $51.95$ & $51.89$ & $51.89$ \\
            GHS-PCA & $51.67$ & $51.67$ & $51.10$ & $51.10$ \\
            \hline
        \end{tabular}
        }
\end{table}
Also for barrier options (barrier at expiry), we notice that GHS and
LA give directions of stratification  that are almost parallel as
illustrated in Table \ref{tab:AnglesBarrier}. This justifies the
approximation of the LA method and its use for stratified MC to
price the two types of barrier options. In addition, the GHS
algorithm is not applicable to Asian options with a complete
barrier. Different approaches should be employed in order to improve
the stratification efficiency for barrier-style options, as
suggested in Etoré et al. \cite{EFJM09}, but these are nevertheless
computationally expensive and use orthogonal directions. The
stratified MC does not return variances as low as for plain Asian
options, especially when the barrier is close to the strike price.
For example, the case of Asian options with barrier $B=80$ (both at
expiry and at each sampling date) and with strike $K_S=55$ displays
a variance reduction of several hundreds with a computational time
that ranges between $22\%$ and $55\%$ higher than the standard MC.
However, when the barrier and the strike price are $K_S=50$ and
$B=60$, respectively, the variance reduction is lower with an extra
effort ranging from $22\%$ and $50\%$  with respect to the standard
MC.

The numerical simulation of the prices of Asian basket options with
a barrier close to the strike price, both at expiry and at all the
monitoring times, shows that stratifying along multiple directions
can be worthwhile. Indeed, if $K_S=50$ and $B=60$, the multiple
stratification enhances the accuracy of the estimation compared to
the use of a single direction. In particular, the highest variance
reduction is achieved with the choice of non-orthogonal directions
(LA-PCA) with optimal allocation. In this setting the variance
reduction is of an order $100$, with barrier at expiry, or $40$,
with barrier at each monitoring time, and is several times higher
compared to the other setting of stratification.

Finally, even for Asian barrier options the LHS never outperforms
the technique that displays the smallest variance with optimal
allocation. These considerations suggest that the use of multiple
non-orthogonal directions can be worthwhile. However, finding many
different multiple directions is not a simple task.
\subsection{Basket Options in the Black-Scholes Market}
 \begin{table}
 \caption{Input Parameters and Angles between Directions of Stratification  for Basket Options.}
    \subtable[Input Parameters.]{\label{tab:BasketCIRPar}
     \begin{tabular}{|c|c|c|c|c|c|}
         \hline
         $M$ & $S_0$ & $\rho$ & $r$ & $\sigma$ & $T$ \\
         \hline
         40 &Linear 20-60 &0.5&0.05& Linear $0.1-0.4$ & $1$ \\
         \hline
     \end{tabular}
     }
\subtable[Angles in degrees]{\label{tab:BasketBSAngles}
            \begin{tabular}{|c|c|c|c|}
            \hline
            &$K_S=30$ & $K_S=40$ & $K_S=50$\\
            \hline
            LA-GHS & $2.76$ &$3.11$& $2.52$ \\
            LA-PCA & $64.74$& $65.04$& $65.19$\\
            GHS-PCA & $62.29$&$62.02$ & $62.47$\\
            \hline
            \end{tabular}
            }
 \end{table}
In this example the stratification estimator once more improves the
accuracy of the standard MC method. Indeed, in the BS market, the
financial features of basket options are almost the same as those of
arithmetic Asian options. The main difference between the two is
that for Asian options the Gaussian variables are correlated by the
autocovariance matrix of a single Brownian motion while for basket
options the dependence is measured by the covariance matrix among
the asset returns. In addition, both payoffs contain a (weighted)
average of the exponential of a Gaussian random vector. Table
\ref{tab:Basket} shows that for all the considered exercise prices,
the stratification using the LA (LT) with and without optimal
allocation has a remarkable variance reduction comparable to the one
given by the GHS algorithm with the same computational
considerations as in the Asian option example. Indeed, these two
directions are almost parallel (see Table \ref{tab:BasketBSAngles}).
The PCA-based direction has again a modest effect in terms of
variance reduction and the stratification over a single linear
projection produces a better accuracy than the one that exploits two
directions. Finally, the LHS estimator neither achieves a higher
variance reduction than the stratified estimator with a single LA
direction with optimal allocation  nor does it require a lower
computational effort.
\subsection{Asian Options in the CIR Market}
 \begin{table}
 \centering \caption{Input Parameters and Angles between Directions of Stratification in the CIR dynamics.}
 \subtable[Input Parameters]{
 \begin{tabular}{|c|c|c|c|c|c|c|}
         \hline
         $S_0$ & $N$ & $r$ & $\alpha$ & $\mu$ & $\sigma$ & $T$ \\
         \hline
         $100$ &$64$&$0.05$& $1.5$ & $1$ & $0.8$ & $1$ \\
         \hline
 \end{tabular}\label{tab:CIRparameters}
 }
 \subtable[Angles in degrees for Asian Options]{
            \begin{tabular}{|c|c|c|c|}
            \hline
            &$K_S=90$ & $K_S=100$ & $K_S=110$\\
            \hline
            LA-LT & $1.00$ &$1.00$& $1.00$ \\
            LA-PCA & $43.72$& $43.72$& $43.72$\\
            LT-PCA & $44.24$&$41.52$ & $41.53$\\
            \hline
            \end{tabular}
            }\label{tab:AsianCIRAngles}
\end{table}
As a last example we consider arithmetic Asian options on a single
asset in a CIR dynamics whose depicted parameters (in Table
\ref{tab:CIRparameters}) are chosen in order to ensure positive
prices ($2\alpha\mu>\sigma^2$). In this setting the LA method and
the LT decomposition do not provide the same stratification
direction and the GHS algorithm is really difficult to apply.
However, as illustrated in Table \ref{tab:AsianCIRAngles} the
directions returned by the LT and LA are almost parallel. In any
case the derivation of the LA solution and its implementation are
much easier. Since the CIR model is neither a Gaussian nor a
lognormal process,  the PCA decomposition is not applicable.
However, in order to obtain a further direction we estimate a
PCA-like direction as explained at the beginning of this section.
Tables \ref{tab:AsianResultsCIR}-\ref{tab:BarrierCompleteResultsCIR}
show that both the LA and LT algorithms give remarkable variance
reductions. The best accuracies are obtained with the stratification
along a single direction which attains a reduction of an order of
several hundreds, both with a constant and optimal allocation rule.
The extra cost for the computational time is only $20\%$. As in the
BS setting, the PCA approach is less efficient and requires a higher
computational cost due to calculation of the sampled autocovariance
matrix of the price process. Also in this situation the solution
employing two orthogonal or  non-orthogonal directions provides a
variance reduction. Unfortunately, this choice never provides an
accuracy as precise as the one obtained by a single direction.
Moreover, the use of the fixed LHS-allocation rule never enhances
the accuracy of the simulation more than the best low-dimensional
stratification method with optimal allocation.

As in the BS example, we add a knock-out barrier at expiry or at
each monitoring time. For this latter option we must  chose a
barrier level that is much higher than the strike price. Indeed, due
to the high variability of the CIR dynamics, with a low barrier
value the option would easily knock-out producing a zero-valued
price.

As already mentioned, in the example of barrier options we adopt the
same convenient directions of stratification that we would consider
without the barrier since the LA and LT approaches do not take into
account the non-differentiable part of the payoff. Tables
\ref{tab:ExpiryResultsBS} and \ref{tab:BarrierCompleteResultsBS}
illustrate the results of this numerical investigation. The variance
reduction is not as efficient as the case without barrier but in
contrast, the use of multiple directions improves the efficiency of
the simulation without highly influencing the computational cost. In
addition, the combination of non-orthogonal directions can achieve a
better variance reduction. Indeed, the combination of LA-PCA
directions (LT and LA are almost parallel) returns a variance that
ranges from $10$ to $30$ times lower than that with standard Monte
Carlo. Moreover, this estimated variance is always at least equal,
for $K_S=100,\,B=170$ with barrier at each monitoring time, or lower
than the variance obtained with different combinations of
stratifying directions and barrier levels.

Finally, as in all examples, the LHS sampling coupled with LT does
not provide a convenient alternative to stratification over few
directions with optimal allocation.
\section{Concluding Remarks and Future Perspectives}\label{sect:concl}
In this paper we have investigated the use of convenient
multidimensional directions of stratification in order to enhance
the accuracy of Monte Carlo methods. We have discussed directions of
stratification  that are easy to derive and   display variance
reductions that are comparable to those introduced by Glasserman et
al. \cite{GHS99}. These solutions do not require a complex
calculation and can be applied in really high-dimensional problems
without an extra cost. In contrast, the use of the Glasserman et al.
\cite{GHS99} or Etoré et al. \cite{EFJM09} methods risk to be
computationally unfeasible and are based only on orthogonal
directions. Indeed, the LT and the LA directions are computed under
convenient approximations that lead to simple matrix operations and
vector norms. Moreover, we have proved an algorithm that allows to
correctly generate Gaussian vectors stratified along non-orthogonal
directions. Our numerical experiments demonstrate that the proposed
convenient directions return remarkable variance reductions both in
BS, where the proposed techniques display the same variance
reduction as those given by GHS, and in the CIR dynamics. In
particular, the use of multiple non-orthogonal directions can be
worthwhile for barrier style options. Moreover, in this work we show
that the use of a few convenient directions of stratification with
optimal allocation always outperform LHS (even in its LT-enhanced
form) especially in terms of computational burden. A natural
extension would be the combination with importance sampling
procedures like the Robust Adaptive Technique recently proposed by
Jourdain and Lelong \cite{JL09} for Gaussian random vectors. In
addition, due to its simple derivation and its affinity with the
Fox's  greedy rule (see Fox \cite{Fox99}), it would be interesting
to investigate how to apply the LA procedure to derive a Quasi-Monte
Carlo version of discretization schemes for stochastic volatility
models like those proposed by Andersen \cite{And2007} and Jourdain
and Sbai \cite{JS09}. \clearpage

\begin{sidewaystable}
\centering
\caption{Results for Arithmetic Asian Options in the BS
dynamics.} \label{tab:AsianResultsBS}
    \footnotesize
\begin{tabular}{|c|c|c|c|c|c|c|c|c|c|c|c|c|c|c|c|c|c|c|c|c|}
    \hline
     & Price &  & & \multicolumn{6}{|c|}{1 Dir }& \multicolumn{10}{|c|}{2 dirs} & \\
     \hline
    & & & MC & \multicolumn{2}{|c|}{GHS} & \multicolumn{2}{|c|}{LA} & \multicolumn{2}{|c|}{PCA}  &  \multicolumn{2}{|c|}{GHS} & \multicolumn{2}{|c|}{LA} & \multicolumn{2}{|c|}{PCA}  & \multicolumn{2}{|c|}{GHS-PCA} & \multicolumn{2}{|c|}{LA-PCA}  & LHS\\
    \hline
    & &  & & const & opt & const & opt & const & opt & const & opt & const & opt & const & opt & const & opt & const & opt & \\
    \hline
    \multirow{2}{*}{$K_S=45$} &\multirow{2}{*}{$7.02$} & var &$55.89$&$0.32$&$0.06$&$0.31$&$0.06$&$15.46$&$11.4$&$1.74$&$0.61$&$0.94$&$0.16$&$10.08$&$8.66$&$8.12$&$0.21$&$8.32$&$0.19$&$0.06$\\
    & & time
    &$1$&$\times1.41$&$\times1.51$&$\times1.41$&$\times1.51$&$\times1.41$&$\times1.51$&$\times1.41$&$\times1.51$&$\times1.41$&$\times1.51$&$\times1.41$&$\times1.51$&$\times1.58$&$\times1.68$&$\times1.58$&$\times1.68$&$\times3.6$\\
    \hline
    \multirow{2}{*}{$K_S=50$} &\multirow{2}{*}{$4.02$} & var &$36.966$&$0.28$&$0.04$&$0.31$&$0.05$&$20.94$&$16.18$&$0.95$&$0.2$&$0.94$&$0.12$&$7.77$&$6.18$&$9.47$&$0.21$&$9.21$&$0.2$&$0.06$\\
    & & time
    &$1$&$\times 1.41$&$\times 1.51$&$\times 1.41$&$\times 1.51$&$\times 1.41$&$\times 1.51$&$\times 1.41$&$\times 1.51$&$\times 1.41$&$\times 1.51$&$\times 1.41$&$\times 1.51$&$\times 1.58$&$\times 1.68$&$\times 1.58$&$\times 1.68$&$\times 3.24$\\
    \hline
    \multirow{2}{*}{$K_S=55$} &\multirow{2}{*}{$2.06$} & var &$20.357$&$0.3$&$0.02$&$0.31$&$0.03$&$10.52$&$7.75$&$1.06$&$0.28$&$0.93$&$0.09$&$7.54$&$3.8641$&$7.4$&$0.13$&$7.49$&$0.13$&$0.06$\\
    & & time
    &1&$\times 1.41$&$\times 1.51$&$\times 1.41$&$\times 1.51$&$\times 1.41$&$\times 1.51$&$\times 1.41$&$\times 1.51$&$\times 1.41$&$\times 1.51$&$\times 1.41$&$\times 1.51$&$\times 1.58$&$\times 1.68$&$\times 1.58$&$\times 1.68$&$\times 3.77$\\
    \hline
\end{tabular}

\caption{Results for Arithmetic Asian Options with a Barrier at
Expiry in the BS dynamics.}\label{tab:ExpiryResultsBS}
    \footnotesize
\begin{tabular}{|c|c|c|c|c|c|c|c|c|c|c|c|c|c|c|c|c|c|c|c|c|}
    \hline
     & Price &  & & \multicolumn{6}{|c|}{1 Dir }& \multicolumn{10}{|c|}{2 dirs} & \\
     \hline
    & & & MC & \multicolumn{2}{|c|}{GHS} & \multicolumn{2}{|c|}{LA} & \multicolumn{2}{|c|}{PCA}  &  \multicolumn{2}{|c|}{GHS} & \multicolumn{2}{|c|}{LA} & \multicolumn{2}{|c|}{PCA}  & \multicolumn{2}{|c|}{GHS-PCA} & \multicolumn{2}{|c|}{LA-PCA}  & LHS\\
    \hline
    & &  & & const & opt & const & opt & const & opt & const & opt & const & opt & const & opt & const & opt & const & opt & \\
    \hline
    \multirow{2}{*}{$\begin{array}{c}
                      K_S=50 \\B=60
                    \end{array}$
    } &\multirow{2}{*}{$1.38$} & var &$2.99$&$1.13$&$0.3$&$1.13$&$0.31$&$2.99$&$2.99$&$0.54$&$0.23$&$0.83$&$0.19$&$1.24$&$0.93$&$0.33$&$0.02$&$0.32$&$0.02$&$1.02$\\
    & & time
    &$1$&$\times 1.41$&$\times 1.41$&$\times 1.41$&$\times 1.35$&$\times 1.41$&$\times 1.35$&$\times 1.47$&$\times 1.47$&$\times 1.47$&$\times 1.40$&$\times 1.47$&$\times 1.40$&$\times 1.50$&$\times 1.50$&$\times 1.55$&$\times 1.50$&$\times 3.91$\\
    \hline
    \multirow{2}{*}{$\begin{array}{c}
                      K_S=50 \\B=70
                    \end{array}$
    } &\multirow{2}{*}{$1.9$} & var &$4.8$&$0.13$&$0.01$&$0.13$&$0.01$&$4.77$&$4.77$&$0.3$&$0.16$&$0.15$&$0.02$&$1.28$&$0.99$&$0.41$&$0.02$&$0.68$&$0.02$&$0.13$\\
    & & time
    &$1$&$\times 1.41$&$\times 1.41$&$\times 1.41$&$\times 1.35$&$\times 1.41$&$\times 1.35$&$\times 1.47$&$\times 1.47$&$\times 1.47$&$\times 1.40$&$\times 1.47$&$\times 1.40$&$\times 1.55$&$\times 1.50$&$\times 1.55$&$\times 1.50$&$\times 3.90$\\
    \hline
    \multirow{2}{*}{$\begin{array}{c}
                      K_S=55 \\B=70
                    \end{array}$} &\multirow{2}{*}{$0.19$} & var &$0.49$&$0.04$&$0.00074$&$0.04$&$0.00082$&$0.48$&$0.48$&$0.04$&$0.0035$&$0.06$&$0.0039$&$0.22$&$0.06$&$0.17$&$0.0038$&$0.16$&$0.0037$&$0.04$\\
    & & time
    &$1$&$\times 1.41$&$\times 1.41$&$\times 1.41$&$\times 1.35$&$\times 1.41$&$\times 1.35$&$\times 1.47$&$\times 1.47$&$\times 1.47$&$\times 1.40$&$\times 1.47$&$\times 1.40$&$\times 1.55$&$\times 1.50$&$\times 1.55$&$\times 1.50$&$\times 3.89$\\
    \hline
        \multirow{2}{*}{$\begin{array}{c}
                      K_S=55 \\B=80
                    \end{array}$
    } &\multirow{2}{*}{$0.2$} & var &$0.55$&$0.0016$&0.00026&$0.0018$&$0.00058$&$0.55$&$0.54$&$0.05$&$0.0037$&$0.06$&$0.0038$&$0.22$&$0.06$&$0.18$&$0.0048$&$017$&$0.0048$&$0.0018$\\
    & & time
    &$1$&$\times 1.41$&$\times 1.41$&$\times 1.41$&$\times 1.35$&$\times 1.41$&$\times 1.35$&$\times 1.47$&$\times 1.47$&$\times 1.47$&$\times 1.40$&$\times 1.47$&$\times 1.40$&$\times 1.50$&$\times 1.55$&$\times 1.50$&$\times 1.55$&$\times 3.91$\\
    \hline
\end{tabular}

\caption{Results for Arithmetic Asian Options with a Complete
Barrier in the BS dynamics.}\label{tab:BarrierCompleteResultsBS}
    \footnotesize
\begin{tabular}{|c|c|c|c|c|c|c|c|c|c|c|c|c|c|c|}
    \hline
     & Price &  & & \multicolumn{4}{|c|}{1 Dir }& \multicolumn{6}{|c|}{2 dirs} & \\
     \hline
    & & & MC &  \multicolumn{2}{|c|}{LA} & \multicolumn{2}{|c|}{PCA}  &   \multicolumn{2}{|c|}{LA} & \multicolumn{2}{|c|}{PCA} & \multicolumn{2}{|c|}{LA-PCA}  & LHS\\
    \hline
    & &  & & const & opt & const & opt & const & opt & const & opt & const & opt &  \\
    \hline
    \multirow{2}{*}{$\begin{array}{c}
                      K_S=50 \\B=60
                    \end{array}$
    } &\multirow{2}{*}{$1.22$} & var &$2.42$&$0.85$&$0.23$&$2.42$&$2.39$&$0.54$&$0.12$&$1.23$&$0.92$&$0.53$&$0.07$&$0.77$\\
    & & time
    &$1$&$\times 1.54$&$\times 1.14$&$\times 1.54$&$\times 1.14$&$\times 1.54$&$\times 1.14$&$\times 1.54$&$\times 1.14$&$\times 1.56$&$\times 1.22$&$\times 3.80$\\
    \hline
    \multirow{2}{*}{$\begin{array}{c}
                      K_S=50 \\B=70
                    \end{array}$
    } &\multirow{2}{*}{$1.89$} & var &$4.76$&$0.14$&$0.0047$&$4.75$&$4.75$&$0.16$&$0.02$&$1.29$&$1$&$1.52$&$0.02$&$0.15$\\
    & & time
    &$11.17$&$\times 1.54$&$\times 1.14$&$\times 1.54$&$\times 1.14$&$\times 1.54$&$\times 1.14$&$\times 1.54$&$\times 1.14$&$\times 1.56$&$\times 1.22$&$\times 3.81$\\
    \hline
    \multirow{2}{*}{$\begin{array}{c}
                      K_S=55 \\B=70
                    \end{array}$} &\multirow{2}{*}{$0.19$} & var &$0.47$&$0.041$&$0.00087$&$0.47$&$0.46$&$0.06$&$0.0038$&$0.22$&$0.06$&$0.14$&$0.0036$&$0.04$\\
    & & time
    &$1$&$\times 1.54$&$\times 1.14$&$\times 1.54$&$\times 1.14$&$\times 1.54$&$\times 1.14$&$\times 1.54$&$\times 1.14$&$\times 1.56$&$\times 1.22$&$\times 3.85$\\
    \hline
        \multirow{2}{*}{$\begin{array}{c}
                      K_S=55 \\B=80
                    \end{array}$
    } &\multirow{2}{*}{$0.2$} & var &$0.55$&$0.0015$&$0.000059$&$0.55$&$0.53$&$0.05$&$0.0038$&$0.22$&$0.06$&$0.056$&$0.0048$&$0.002$\\
    & & time
    &$1$&$\times 1.54$&$\times 1.14$&$\times 1.54$&$\times 1.14$&$\times 1.54$&$\times 1.14$&$\times 1.54$&$\times 1.14$&$\times 1.56$&$\times 1.22$&$\times 3.83$\\
    \hline
\end{tabular}

 \centering
\caption{Results for Basket Options in the BS
dynamics.}\label{tab:Basket}
    \footnotesize
\begin{tabular}{|c|c|c|c|c|c|c|c|c|c|c|c|c|c|c|c|c|c|c|c|c|}
    \hline
     & Price &  & & \multicolumn{6}{|c|}{1 Dir }& \multicolumn{10}{|c|}{2 dirs} & \\
     \hline
    & & & MC & \multicolumn{2}{|c|}{GHS} & \multicolumn{2}{|c|}{LA} & \multicolumn{2}{|c|}{PCA}  &  \multicolumn{2}{|c|}{GHS} & \multicolumn{2}{|c|}{LA} & \multicolumn{2}{|c|}{PCA}  & \multicolumn{2}{|c|}{GHS-PCA} & \multicolumn{2}{|c|}{LA-PCA}  & LHS\\
    \hline
    & &  & & const & opt & const & opt & const & opt & const & opt & const & opt & const & opt & const & opt & const & opt & \\
    \hline
    \multirow{2}{*}{$K_S=30$} &\multirow{2}{*}{$11.58$} & var &$61.77$&$0.09$&$0.06$&$0.1$&$0.06$&$31.54$&$21.63$&$0.93$&$0.29$&$0.91$&$0.25427$&$21.17$&$18.34$&$6.33$&$0.24$&$5.27$&$0.25406$&$0.06$\\
    & & time
    &$1$&$\times 1.48$&$\times 1.60$&$\times 1.48$&$\times 1.60$&$\times 1.48$&$\times 1.60$&$\times 1.48$&$\times 1.60$&$\times 1.48$&$\times 1.60$&$\times 1.48$&$\times 1.60$&$\times 1.75$&$\times 1.79$&$\times 1.75$&$\times 1.79$&$2.87$\\
    \hline
    \multirow{2}{*}{$K_S=40$} &\multirow{2}{*}{$4.15$} & var &$34.88$&$0.07$&$0.03$&$0.08$&$0.04$&$24.91$&$17.74$&$0.84$&$0.15$&$0.86$&$0.15$&$19.1$&$16.71$&$3.9$&$0.12$&$3.69$&$0.13214$&$0.1$\\
    & & time
    &$1$&$\times 1.48$&$\times 1.60$&$\times 1.48$&$\times 1.60$&$\times 1.48$&$\times 1.60$&$\times 1.48$&$\times 1.60$&$\times 1.48$&$\times 1.60$&$\times 1.48$&$\times 1.60$&$\times 1.75$&$\times 1.79$&$\times 1.75$&$\times 1.79$&$2.81$\\
    \hline
    \multirow{2}{*}{$K_S=50$} &\multirow{2}{*}{$0.93$} & var &$8.92$&$0.04$&$0.004$&$0.05$&$0.005$&$3.92$&$3.88$&$0.8$&$0.06$&$0.81287$&$0.06$&$3.05$&$2.18$&$2.87$&$0.04$&$2.55$&$0.05$&$0.08$\\
    & & time
    &$1$&$\times 1.48$&$\times 1.60$&$\times 1.48$&$\times 1.60$&$\times 1.48$&$\times 1.60$&$\times 1.48$&$\times 1.60$&$\times 1.48$&$\times 1.60$&$\times 1.48$&$\times 1.60$&$\times 1.75$&$\times 1.79$&$\times 1.75$&$\times 1.79$&$2.89$\\
    \hline
\end{tabular}

\end{sidewaystable}
\clearpage
\newpage
\begin{sidewaystable}
\centering \caption{Results for Asian Options in the CIR
dynamics.}\label{tab:AsianResultsCIR}
    \footnotesize
\begin{tabular}{|c|c|c|c|c|c|c|c|c|c|c|c|c|c|c|c|c|}
    \hline
     & Price &  & & \multicolumn{6}{|c|}{1 Dir }& \multicolumn{6}{|c|}{2 dirs} & \\
     \hline
    & & & MC &  \multicolumn{2}{|c|}{LT} & \multicolumn{2}{|c|}{LA}  &   \multicolumn{2}{|c|}{PCA} & \multicolumn{2}{|c|}{LT} & \multicolumn{2}{|c|}{PCA}  & \multicolumn{2}{|c|}{LA-PCA}  & LHS\\
    \hline
    & &  & & const & opt & const & opt & const & opt & const & opt & const & opt & const & opt &  \\
    \hline
    \multirow{2}{*}{$\begin{array}{c}
                      K_S=90
                    \end{array}$
    } &\multirow{2}{*}{$15.63$} & var &$427.73$&$1.85$&$1.09$&$1.54$&$0.9$&$115.73$&$106.85$&$9.3$&$2.28$&$51.21$&$40.61$&$9.13$&$4.62$&$1.08$\\
    & & time
    &$1$&$\times 1.2$&$\times 1.22$&$\times 1.2$&$\times 1.22$&$\times 1.5$&$\times 1.6$&$\times 1.2$&$\times 1.22$&$\times 1.5$&$\times 1.6$&$\times 1.55$&$\times 1.55$&$\times 2.76$\\
    \hline
    \multirow{2}{*}{$\begin{array}{c}
                      K_S=100
                    \end{array}$
    } &\multirow{2}{*}{$10.6$} & var &$310.11$&$1.49$&$0.67$&1.22&$0.54$&$97.22$&$69.7$&$8.75$&$1.73$&$53.03$&$25.73$&$8.92$&$1.66$&$1.02$\\
    & & time
    &$1$&$\times 1.2$&$\times 1.22$&$\times 1.2$&$\times 1.22$&$\times 1.5$&$\times 1.6$&$\times 1.2$&$\times 1.22$&$\times 1.5$&$\times 1.6$&$\times 1.55$&$\times 1.554$&$\times 2.75$\\
    \hline
    \multirow{2}{*}{$\begin{array}{c}
                      K_S=110
                    \end{array}$} &\multirow{2}{*}{$6.95$} & var &$212.19$&$1.18$&$0.37$&&$0.29$&$82.25$&$54.28$&$8.72$&$1.26$&$40.34$&$20.69$&$8.29$&$2.22$&$0.9$\\
    & & time
    &$1$&$\times 1.2$&$\times 1.22$&$\times 1.2$&$\times 1.22$&$\times 1.5$&$\times 1.6$&$\times 1.2$&$\times 1.22$&$\times 1.5$&$\times 1.6$&$\times 1.55$&$\times 1.55$&$\times 2.76$\\
    \hline
\end{tabular}

\caption{Results for Arithmetic Asian Options with a Barrier at
Expiry in the CIR dynamics.}\label{tab:ExpiryResultsCIR}
    \footnotesize
\begin{tabular}{|c|c|c|c|c|c|c|c|c|c|c|c|c|c|c|c|c|}
    \hline
     & Price &  & & \multicolumn{6}{|c|}{1 Dir }& \multicolumn{6}{|c|}{2 dirs} & \\
     \hline
    & & & MC &  \multicolumn{2}{|c|}{LT} & \multicolumn{2}{|c|}{LA}  &   \multicolumn{2}{|c|}{PCA} & \multicolumn{2}{|c|}{LT} & \multicolumn{2}{|c|}{PCA}  & \multicolumn{2}{|c|}{LA-PCA}  & LHS\\
    \hline
    & &  & & const & opt & const & opt & const & opt & const & opt & const & opt & const & opt &  \\
    \hline
    \multirow{2}{*}{$\begin{array}{c}
                      K_S=100\\B=110
                    \end{array}$
    } &\multirow{2}{*}{$2.63$} & var &$60.43$&$45.76$&$17.77$&45.78&17.22&$55.81$&$38.69$&$26.19$&$9.17$&$40.61$&$12.49$&$20.23$&$3.08$&$39.46$\\
    & & time
    &$1$&$\times 1.2$&$\times 1.22$&$\times 1.2$&$\times 1.22$&$\times 1.5$&$\times 1.6$&$\times 1.2$&$\times 1.22$&$\times 1.5$&$\times 1.6$&$\times 1.55$&$\times 1.55$&$\times 2.91$\\
    \hline
    \multirow{2}{*}{$\begin{array}{c}
                      K_S=110\\B=120
                    \end{array}$
    } &\multirow{2}{*}{$1.82$} & var &$41.64$&$32.64$&$8.1$&$32.55$&$7.85$&$38.69$&$26.43$&$20.76$&$5.77$&$26.4$&$6.27$&$11.95$&$1.26$&$28.52$\\
    & & time
    &$1$&$\times 1.2$&$\times 1.22$&$\times 1.2$&$\times 1.22$&$\times 1.5$&$\times 1.6$&$\times 1.2$&$\times 1.22$&$\times 1.5$&$\times 1.6$&$\times 1.55$&$\times 1.55$&$\times 2.87$\\
    \hline
    \multirow{2}{*}{$\begin{array}{c}
                      K_S=100\\B=120
                    \end{array}$} &\multirow{2}{*}{$3.46$} & var &$81.21$&$34.54$&$20.77$&$31.61$&$20.3$&$53.62$&$33.82$&$37.05$&$15.01$&$50.05$&$15.57$&$21.19$&$4.5$&$48.97$\\
    & & time
    &$1$&$\times 1.2$&$\times 1.22$&$\times 1.2$&$\times 1.22$&$\times 1.5$&$\times 1.6$&$\times 1.2$&$\times 1.22$&$\times 1.5$&$\times 1.6$&$\times 1.55$&$\times 1.55$&$\times 2.87$\\
    \hline
\end{tabular}

\caption{Results for Arithmetic Asian Options with a Complete
Barrier in the CIR dynamics.}\label{tab:BarrierCompleteResultsCIR}
    \footnotesize
\begin{tabular}{|c|c|c|c|c|c|c|c|c|c|c|c|c|c|c|c|c|}
    \hline
     & Price &  & & \multicolumn{6}{|c|}{1 Dir }& \multicolumn{6}{|c|}{2 dirs} & \\
     \hline
    & & & MC &  \multicolumn{2}{|c|}{LT} & \multicolumn{2}{|c|}{LA}  &   \multicolumn{2}{|c|}{PCA} & \multicolumn{2}{|c|}{LT} & \multicolumn{2}{|c|}{PCA}  & \multicolumn{2}{|c|}{LA-PCA}  & LHS\\
    \hline
    & &  & & const & opt & const & opt & const & opt & const & opt & const & opt & const & opt &  \\
    \hline
    \multirow{2}{*}{$\begin{array}{c}
                      K_S=100\\B=180
                    \end{array}$
    } &\multirow{2}{*}{$2.84$} & var &$42.98$&$25.39$&$7.44$&$25.38$&$7.32$&$37.52$&$27.98$&$22.4$&$6.25$&$30.14$&$16.5$&$21.37$&$5.57$&$22.58$\\
    & & time
    &$1$&$\times 1.21$&$\times 1.1$&$\times 1.21$&$\times 1.1$&$\times 1.33$&$\times 1.23$&$\times 1.21$&$\times 1.1$&$\times 1.33$&$\times 1.23$&$\times 1.33$&$\times 1.23$&$\times 2.82$\\
    \hline
    \multirow{2}{*}{$\begin{array}{c}
                      K_S=110\\B=180
                    \end{array}$
    } &\multirow{2}{*}{$1.1$} & var &$14.03$&$9.51$&$2.05$&$9.59$&$2.05$&$12.68$&$8.37$&$8.63$&$1.73$&$11.33$&$4.76$&$8.35$&$1.58$&$8.79$\\
    & & time
    &$1$&$\times 1.21$&$\times 1.1$&$\times 1.21$&$\times 1.1$&$\times 1.33$&$\times 1.23$&$\times 1.21$&$\times 1.1$&$\times 1.33$&$\times 1.23$&$\times 1.33$&$\times 1.23$&$\times 2.85$ \\
    \hline
    \multirow{2}{*}{$\begin{array}{c}
                      K_S=100\\B=170
                    \end{array}$} &\multirow{2}{*}{$1.79$} & var &$23.7$&$15.86$&$4.65$&$15.96$&$4.58$&$21.59$&$15.21$&$10.75$&$3.44$&$15.97$&$8.75$&$13.73$&$3.47$&$15.08$\\
    & & time
    &$1$&$\times 1.21$&$\times 1.1$&$\times 1.21$&$\times 1.1$&$\times 1.33$&$\times 1.23$&$\times 1.21$&$\times 1.1$&$\times 1.33$&$\times 1.23$&$\times 1.33$&$\times 1.23$&$\times 2.79$\\
    \hline
\end{tabular}

\end{sidewaystable}
\clearpage
\newpage
\bibliographystyle{plain}
\bibliography{tesi}

\begin{thebibliography}{10}

\bibitem{AL1998}
F.~{\AA}kesson and J.P. Lehoczky.
\newblock Discrete {E}igenfuction {E}xpansion of {M}ulti-{D}imensional
  {B}rownian {M}otion and the {O}rnstein-{U}hlenbeck {P}rocess.
\newblock Technical Report, 1998.

\bibitem{And2007}
L.~Andersen.
\newblock Efficient {S}imulation of the {H}eston {S}tochastic {V}olatility
  {M}odel.
\newblock Available in www.ssrn.com, 2007.

\bibitem{CMO1997}
R.~Caflisch, W.~Morokoff, and A.~Owen.
\newblock {V}aluation of {M}ortgage-backed {S}ecurities {U}sing {B}rownian
  {B}ridges to {R}educe {E}ffective {D}imension.
\newblock {\em Journal of Computational Finance}, pages 27--46, 1997.

\bibitem{EFJM09}
P.~Etor\'e, G.~Fort, B.~Jourdain, and E.~Moulines.
\newblock On {A}daptive {S}tratification.
\newblock Forthcoming in Annals of Operations Research.

\bibitem{EJ09}
P.~Etor\'e and B.~Jourdain.
\newblock Adaptive {O}ptimal {A}llocation in {S}tratified {S}ampling {M}ethods.
\newblock Forthcoming in \textit{Methodology and Computing in Applied
  Probability}.

\bibitem{Fox99}
B.L. Fox.
\newblock {\em Strategies for {Q}uasi-{M}onte {C}arlo}.
\newblock Kluwer Academic Publishers, 1999.

\bibitem{Glass2004}
P.~Glasserman.
\newblock {\em Monte {C}arlo {M}ethods in {F}inancial {E}ngineering}.
\newblock Springer-Verlag New York, 2004.

\bibitem{GHS99}
P.~Glasserman, P.~Heidelberger, and P.~Shahabuddin.
\newblock Asymptotically {O}ptimal {I}mportance {S}ampling and {S}tratification
  for {P}ricing {P}ath-dependent {O}ptions.
\newblock {\em Mathematical Finance}, pages 117--152, 1999.

\bibitem{IT2006}
J.~Imai and K.S. Tan.
\newblock A {G}eneral {D}imension {R}eduction {T}echnique for {D}erivative
  {P}ricing.
\newblock {\em Journal of Computational Finance}, pages 129--155, 2006.

\bibitem{JL09}
B.~Jourdain and J.~Lelong.
\newblock Robust {A}daptive {I}mportance {S}ampling for {N}ormal {R}andom
  {V}ectors.
\newblock {\em Annals of Applied Probability}, pages 1687--1718, 2009.

\bibitem{JS09}
B.~Jourdain and M.~Sbai.
\newblock High {O}rder {D}iscretization {S}chemes for {S}tochastic {V}olatility
  {M}odels.
\newblock Preprint arXiv 0908-1926, 2009.

\bibitem{LL96}
D.~Lamberton and B.~Lapeyre.
\newblock {\em {I}ntroduction to {S}tochastic {C}alculus {A}pplied to
  {F}inance}.
\newblock Chapman \& Hall, 1996.

\bibitem{ow1992a}
A.~Owen.
\newblock A {C}entral {L}imit {T}heorem for {L}atin {H}ypercube {S}ampling.
\newblock {\em Journal of the Royal Statistical Society}, pages 541--551, 1992.
\newblock Series B (Methodological).

\bibitem{Pap02}
A.~Papageorgiou.
\newblock The {B}rownian {B}ridge {D}oes {N}ot {O}ffer a {C}onsistent
  {A}vantage in {Q}uasi-{M}onte {C}arlo {I}ntegration.
\newblock {\em Journal of Complexity}, 18:171--186, 2002.

\bibitem{Sabino08b}
P.~Sabino.
\newblock Efficient {Q}uasi-{M}onte {S}imulations for {P}ricing
  {H}igh-dimensional {P}ath-dependent {O}ptions.
\newblock {\em Decision in {E}conomics and {F}inance}, 32(1):48--65, 2009.

\bibitem{Sab08}
P.~Sabino.
\newblock Implementing {Q}uasi-{M}onte {C}arlo {S}imulations with {L}inear
  {T}ransformations.
\newblock \emph{Computational Management Science}, in press., 2009.

\bibitem{Stein87}
M.~Stein.
\newblock Large {S}ample {P}roperties of {S}imulations {U}sing {L}atin
  {H}ypercube {S}ampling.
\newblock {\em Technometrics}, pages 143--51, 1987.

\end{thebibliography}
\end{document}